\documentclass[journal]{IEEEtran}

\usepackage[a4paper, left=1.55cm, right=1.55cm, top=1.5cm, bottom=1.4cm]{geometry}  

\usepackage{amsmath}
\usepackage{amssymb}
\usepackage{amsthm}
\usepackage{tikz}
\usetikzlibrary{arrows.meta,positioning,fit,calc}
\theoremstyle{plain}
\newtheorem{theorem}{Theorem}       
\newtheorem{lemma}{Lemma}           
\newtheorem{proposition}{Proposition} 

\theoremstyle{definition}
\newtheorem{definition}{Definition} 

\theoremstyle{remark}
\newtheorem{remark}{Remark}         

\usepackage{graphicx}
\usepackage{url}
\usepackage{booktabs}       
\usepackage{multirow}
\usepackage{makecell}
\usepackage{array}
\usepackage{threeparttable}
\usepackage{algorithm}
\usepackage{algorithmic}
\usepackage{subfig}
\usepackage{float}
\usepackage{rotating}
\usepackage{adjustbox}
\usepackage{xcolor}
\usepackage{cite}
\usepackage{balance}
\usepackage{changepage}
\usepackage{etoolbox}

\begin{document}

	\title{FGFRFT: Fast Graph Fractional Fourier Transform via Exact Spectral Splitting and Fourier-Series Approximation}
	
	\author{Ziqi~Yan, Mingzhi~Wang, Sen~Shi, Feiyue~Zhao, Manjun~Cui, Yangfan~He, and Zhichao~Zhang,~\IEEEmembership{Member,~IEEE}
		\thanks{This work was supported in part by the Open Foundation of Hubei Key Laboratory of Applied Mathematics (Hubei University) under Grant HBAM202404; in part by the Foundation of Key Laboratory of System Control and Information Processing, Ministry of Education under Grant Scip20240121; and in part by the Startup Foundation for Introducing Talent of Nanjing Institute of Technology under Grant YKJ202214. \emph{(Corresponding author Zhichao~Zhang.)}}
		\thanks{Ziqi~Yan, Mingzhi~Wang, Sen~Shi, Feiyue~Zhao, and Manjun~Cui are with the School of Mathematics and Statistics, Nanjing University of Information Science and Technology, Nanjing 210044, China (e-mail yanziqi54@gmail.com;wmz200208@163.com; 202312380032@nuist.edu.cn; 202511150010@nuist.edu.cn; cmj1109@163.com).}
		\thanks{Yangfan~He is with the School of Communication and Artificial Intelligence, School of Integrated Circuits, Nanjing Institute of Technology, Nanjing 211167, China, and also with the Jiangsu Province Engineering Research Center of IntelliSense Technology and System, Nanjing 211167, China (e-mail Yangfan.He@njit.edu.cn).}
		\thanks{Zhichao~Zhang is with the School of Mathematics and Statistics, Nanjing University of Information Science and Technology, Nanjing 210044, China, with the Hubei Key Laboratory of Applied Mathematics, Hubei University, Wuhan 430062, China, and also with the Key Laboratory of System Control and Information Processing, Ministry of Education, Shanghai Jiao Tong University, Shanghai 200240, China (e-mail zzc910731@163.com).}}

	\maketitle

\begin{abstract}
	The graph fractional Fourier transform (GFRFT) for unitary graph Fourier transform (GFT) matrices can be interpreted through the scalar function $e^{j\alpha\theta}$ on the unit circle. Under the principal branch, its Fourier-series representation encounters an intrinsic obstruction at the spectral point $\lambda=-1$ for non-integer orders. To address this issue, we propose a fast graph fractional Fourier transform (FGFRFT) based on exact spectral splitting: the $\lambda=-1$ component is treated exactly, and the complementary component is approximated by a truncated Fourier series in integer powers of the GFT matrix. This construction yields an offline--online implementation that reduces the online complexity of repeated operator updates from $O(N^3)$ to $O(2LN^2)$ for truncation order $L$, while preserving differentiability with respect to the transform order. We further derive truncation-error bounds, approximate unitarity and  additivity, and reconstruction-error bounds. Experiments on approximation accuracy, transform-order learning, image denoising, and point-cloud denoising show that FGFRFT provides substantial online acceleration while remaining close to the exact GFRFT under the tested settings.
\end{abstract}

	\section{Introduction}

	\IEEEPARstart{I}{n} recent years, graph signal processing (GSP) has emerged as a powerful theoretical framework for analyzing high-dimensional data in irregular domains, such as social networks, biological systems, and sensor data \cite{ref1,ref2,ref3,ref4}. It has been widely applied to tasks such as filtering, sampling, and reconstruction \cite{ref5,ref6,ref7,ref8,ref9,ref10,ref11}. As the cornerstone of GSP, the graph Fourier transform (GFT) is the fundamental tool for transforming graph signals from the vertex domain to the spectral domain. Building on this foundation, researchers have generalized the GFT to the graph fractional Fourier transform (GFRFT) by introducing the concept of transform order. The GFRFT can transform graph signals from the vertex domain to an intermediate domain between the vertex and spectral domains, thereby offering a more flexible approach for the representation and processing of graph signals \cite{ref12,ref13,ref14,ref15,ref16,ref17,ref18}.
	
Wang et al. first established the mathematical framework of the GFRFT based on the eigendecomposition of the graph Laplacian matrix \cite{ref19}. Since then, GFRFT has shown advantages in sampling and reconstruction, filtering and detection, data compression, and graph-signal encryption \cite{ref20,ref21,ref22,ref23,ref24}. More recently, adaptive and learnable variants have been developed to move beyond grid-search-based order selection, including adaptive order updates, trainable joint time-vertex transforms, and fractional filtering models for denoising, classification, and graph-analytic signal construction \cite{ref25,ref26,ref27,ref28,ref29}. Despite these advances, repeated order selection or order updates still require expensive eigendecomposition and dense matrix operations, which creates a major computational bottleneck and limits the scalability of GFRFT-based methods in repeated-query and learning scenarios.

To address this issue, we revisit repeated GFRFT order updates from a spectral-analytic perspective. For unitary GFT matrices, the GFRFT can be characterized by the scalar function $e^{j\alpha\theta}$ on the unit circle, which suggests a Fourier-series-based operator representation. Under the principal branch, however, the non-integer-order case exhibits an intrinsic obstruction at the spectral point $\lambda=-1$, where the standard sinc-series construction fails to reproduce the exact GFRFT on the associated eigenspace. We therefore isolate the $\lambda=-1$ component and treat it exactly, while approximating the complementary component by a truncated Fourier series in integer powers of the GFT matrix. This leads to FGFRFT, a structured offline--online framework for repeated order updates under unitary GFT matrices. Computationally, the proposed method reduces the online complexity of dense operator construction from $O(N^3)$ to $O(2LN^2)$ for truncation order $L$, while preserving differentiability with respect to the transform order and compatibility with gradient-based optimization.

The main contributions of this paper are as follows
\begin{itemize}
	\item We reveal and resolve the principal-branch obstruction at $\lambda=-1$ in the Fourier-series representation of non-integer-order GFRFT.
	\item We propose FGFRFT, an offline--online Fourier-series framework for repeated GFRFT order updates under unitary GFT matrices, with online complexity $O(2LN^2)$.
	\item We provide theoretical analysis of FGFRFT, including error bounds, adjoint symmetry, differentiability, approximate unitarity and  additivity, and reconstruction guarantees.
	\item We verify the proposed method through approximation, robustness, and downstream learning and denoising experiments.
\end{itemize}
    
	The remainder of this paper is organized as follows. Section~II reviews the preliminaries. Section~III presents the proposed FGFRFT and its theoretical properties. Section~IV analyzes the approximation accuracy and computational efficiency of the method. Section~V reports the experimental results. Section~VI concludes this paper. The overall framework is illustrated in Fig. \ref{figjishu}.
	
	\begin{figure*}[htbp]
		\centering
		\includegraphics[width=0.8\linewidth]{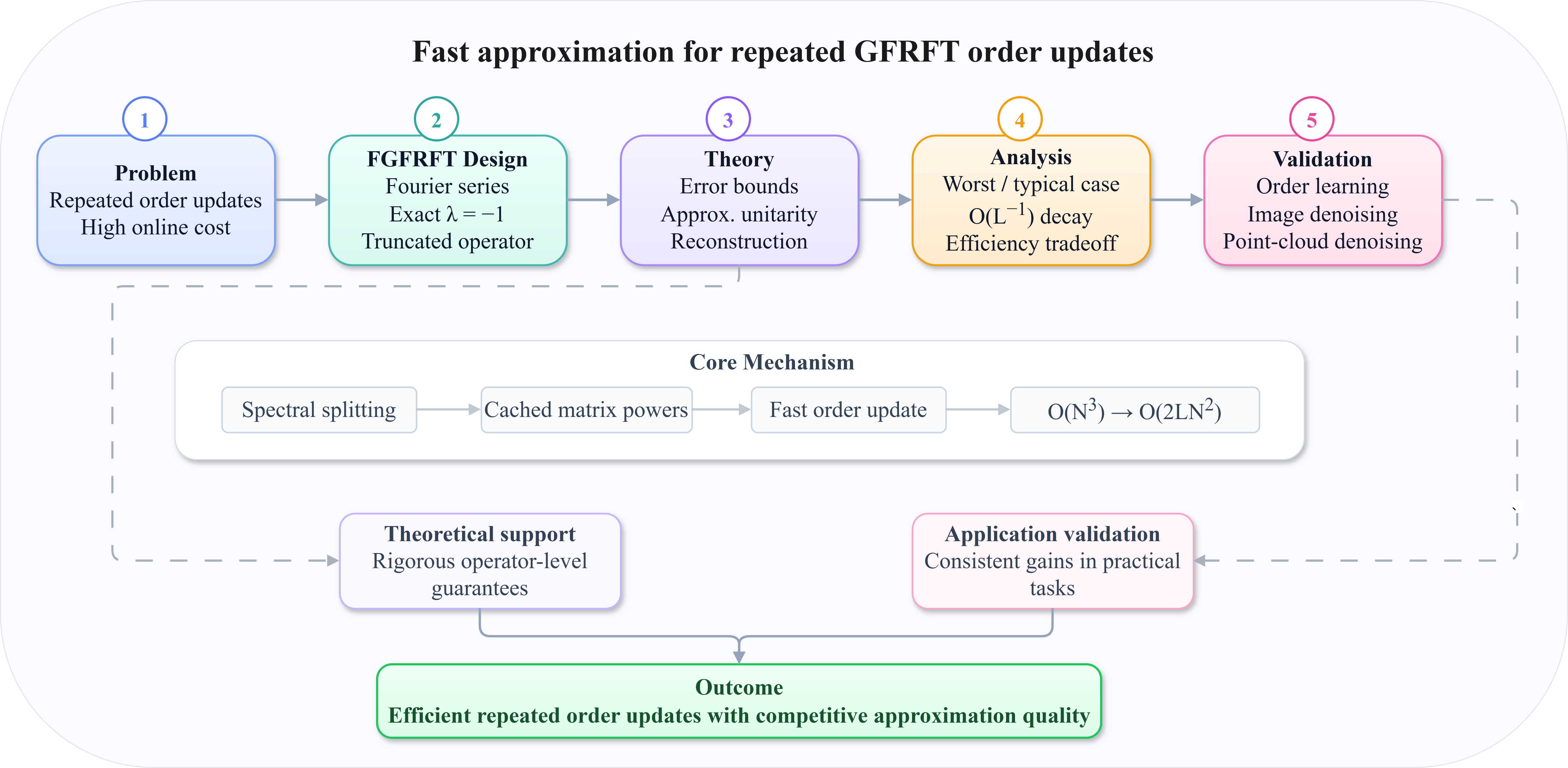} 
		\caption{Overall framework of the paper}
		\label{figjishu} 
	\end{figure*}

\textit{Notation}:
Scalars are denoted by lowercase letters, vectors by bold lowercase letters, and matrices by bold uppercase letters. $(\cdot)^H$ denotes the Hermitian transpose, $\|\cdot\|_2$ and $\|\cdot\|_F$ denote the spectral norm and Frobenius norm, respectively, and $\mathrm{diag}(\cdot)$ denotes a diagonal matrix formed from its arguments. The identity matrix is denoted by $\mathbf I$.
\section{Preliminaries}

\subsection{Graph Fourier Transform}
We define a graph structure as $\mathcal{G} = (\mathcal{V}, \mathbf{A})$, where $\mathcal{V} = \{v_1, \dots, v_N\}$ contains $N$ nodes. The weighted adjacency matrix $\mathbf{A} \in \mathbb{R}^{N \times N}$ describes the connectivity relationships between nodes, where its element $\mathbf{A}_{m,n}$ is non-zero if and only if there exists an edge between node $n$ and node $m$; if $\mathbf{A}_{m,n} = 0$, it indicates no connection. Typically, if $\mathbf{A}_{m,n} = \mathbf{A}_{n,m}$ is satisfied, the graph is considered an undirected graph \cite{ref31}.
The definition of the GFT relies on the spectral decomposition of the graph shift operator (GSO). Let $\mathbf{Z} \in \mathbb{R}^{N \times N}$ be a symmetric GSO of arbitrary form; common choices include the adjacency matrix $\mathbf{A}$, the Laplacian matrix $\mathbf{L}$, and their normalized forms \cite{ref32,ref33,ref34}. The decomposition of $\mathbf{Z}$ can be obtained as \cite{ref19}
\begin{equation}
	\mathbf{Z} = \mathbf{U_Z}\mathbf{\Sigma_Z}\mathbf{U_Z}^{H},
\end{equation}
where $\mathbf{\Sigma_Z}$ is the diagonal eigenvalue matrix, and the column vectors of $\mathbf{U_Z} = [\mathbf{u}_1, \dots, \mathbf{u}_N]$ constitute the eigenvectors of $\mathbf{Z}$. Based on this, the GFT matrix is defined as the inverse of the eigenvector matrix, i.e., $\mathbf{F}_G = \mathbf{U_Z}^{H}$.

For any graph signal $\mathbf{x}$, its spectral representation $\tilde{\mathbf{x}}$ in the GFT domain can be calculated via the following formula
\begin{equation}
	\tilde{\mathbf{x}} \triangleq \mathbf{F}_G \mathbf{x} = \mathbf{U_Z}^{H} \mathbf{x}.
\end{equation}

Correspondingly, the inverse graph Fourier transform (IGFT) is used to recover the graph signal from the spectral domain to the vertex domain
\begin{equation}
	\mathbf{x} = \mathbf{F}_G^{-1} \tilde{\mathbf{x}} = \mathbf{U_Z}\tilde{\mathbf{x}}.
\end{equation}

\subsection{Graph Fractional Fourier Transform}
The GFRFT is a generalization of the traditional GFT, aiming to introduce the theory of fractional Fourier transform from classical signal processing into the graph signal domain. Given an arbitrary graph $\mathcal{G}$ and its symmetric GSO matrix $\mathbf{Z}$, we first consider the eigendecomposition form of the unitary GFT matrix $\mathbf{F}_G$\cite{ref25}
\begin{equation}
	\mathbf{F}_G = \mathbf{U}_{\mathbf{Z}}^{H} = \mathbf{V}\mathbf{\Sigma}\mathbf{V}^H,
\end{equation}
where $\mathbf{\Sigma}$ is the diagonal eigenvalue matrix of $\mathbf{F}_G$. The GFRFT matrix with order $\alpha \in \mathbb{R}$ can be defined as
\begin{equation}
	\mathbf{F}_G^\alpha = \mathbf{V}\mathbf{\Sigma}^\alpha\mathbf{V}^H.
\end{equation}

The above definition satisfy index additivity and possess desirable boundary properties degenerating to the identity matrix $\mathbf{I}$ when $\alpha=0$, and becoming the standard GFT matrix $\mathbf{F}_G$ when $\alpha=1$.

\subsection{Fourier Series Approximation}

The Fourier series is a core mathematical tool for analyzing periodic functions \cite{ref35,ref36,ref37,ref38}. For a $2\pi$-periodic function $g \in L^2([-\pi,\pi])$, its Fourier series is written as
\begin{equation}
	g(\theta) \sim \sum_{n=-\infty}^{\infty} c_n e^{j n \theta}, 
	\qquad \theta \in [-\pi,\pi],
\end{equation}
where the Fourier coefficients are
\begin{equation}
	c_n = \frac{1}{2\pi}\int_{-\pi}^{\pi} g(\theta)e^{-j n \theta}\, d\theta.
\end{equation}

If, in addition, $g$ satisfies the Dirichlet conditions, then its Fourier series converges pointwise to $g(\theta)$ at every continuity point and to
$\frac{g(\theta^-)+g(\theta^+)}{2}$
at every jump point. Uniform convergence generally requires stronger regularity and is not assumed here.
\section{FGFRFT}

This section develops a Fourier-series-based approximation of the GFRFT matrix for the case where the GFT matrix is unitary. A key issue in this construction is the spectral point $\lambda=-1$, which becomes an intrinsic singular point of the sinc-series approximation under the principal branch for non-integer orders. We therefore begin by identifying the associated obstruction, then isolate the corresponding spectral component and treat it exactly, and finally establish the basic properties, approximation error bounds, reconstruction guarantees, and computational complexity of the resulting fast GFRFT (FGFRFT).

\subsection{Spectral Representation and the Singular Point $\lambda=-1$}

Let $\mathbf{F}_{G}\in\mathbb{C}^{N\times N}$ be a unitary GFT matrix. Then
\begin{equation}
	\mathbf{F}_{G}
	=
	\mathbf{V}\mathbf{\Sigma}\mathbf{V}^{H}
	=
	\mathbf{V}\,\mathrm{diag}\!\bigl(e^{j\theta_{1}},\dots,e^{j\theta_{N}}\bigr)\mathbf{V}^{H},
\end{equation}
where $\mathbf{V}$ is unitary and the eigenphases are chosen in the half-open interval $\theta_{k}\in(-\pi,\pi], \qquad k=1,\dots,N$.

Under the principal branch, the GFRFT matrix of order $\alpha\in\mathbb{R}$ is defined by
\begin{equation}
	\mathbf{F}_{G}^{\alpha}
	=
	\mathbf{V}\,\mathrm{diag}\!\bigl(e^{j\alpha\theta_{1}},\dots,e^{j\alpha\theta_{N}}\bigr)\mathbf{V}^{H}.
\end{equation}
Accordingly, the scalar generating function is $g_{\alpha}(\theta)=e^{j\alpha\theta}$.

For $\alpha\in\mathbb{R}$, the $2\pi$-periodic Fourier coefficients of $g_{\alpha}$ are
\begin{equation}
	c_{n}(\alpha)
	=
	\frac{1}{2\pi}\int_{-\pi}^{\pi} e^{j(\alpha-n)\theta}\,d\theta
	=
	\frac{\sin\bigl(\pi(\alpha-n)\bigr)}{\pi(\alpha-n)},
	\qquad n\in\mathbb{Z}.
\end{equation}
Hence, for $\alpha\notin\mathbb{Z}$, $\sum_{n=-\infty}^{\infty} c_{n}(\alpha)e^{jn\theta}$ is the Fourier series associated with the $2\pi$-periodic extension of $g_{\alpha}(\theta)=e^{j\alpha\theta}$.

\begin{remark}[Intrinsic obstruction at $\lambda=-1$]
	Let $\alpha\in\mathbb{R}\setminus\mathbb{Z}$. Then the Fourier series $\sum_{n=-\infty}^{\infty} c_{n}(\alpha)e^{jn\theta}$
	converges to $e^{j\alpha\theta}$ for every $\theta\in(-\pi,\pi)$, but at the endpoint $\theta=\pi$ it converges to the midpoint value
	\begin{equation}
		\frac{e^{j\pi\alpha}+e^{-j\pi\alpha}}{2}
		=
		\cos(\pi\alpha),
	\end{equation}
	rather than to the principal-branch value $e^{j\pi\alpha}$. Consequently, the original sinc-series representation cannot reproduce $\mathbf{F}_{G}^{\alpha}$ on any eigenspace associated with the eigenvalue $\lambda=-1$ for non-integer $\alpha$.
\end{remark}

\begin{remark}[Integer-order case]
	When $\alpha\in\mathbb{Z}$, the Fourier coefficients reduce to the Kronecker delta sequence $c_{n}(\alpha)=\delta_{n,\alpha}$,
	and hence the Fourier-series representation collapses exactly to
		$\sum_{n=-\infty}^{\infty} c_{n}(\alpha)\mathbf{F}_{G}^{n}
		=
		\mathbf{F}_{G}^{\alpha}$.
	Therefore, no truncation or approximation is needed for integer orders. In what follows, the FGFRFT approximation framework is developed only for $\alpha\in\mathbb{R}\setminus\mathbb{Z}$.
\end{remark}

\subsection{Exact Spectral Splitting and FGFRFT Definition}

The remark above shows that the spectral point $\lambda=-1$ is an intrinsic singularity of the sinc-series approximation under the principal branch. We now isolate the spectral component associated with $\lambda=-1$ and treat it exactly.

Define the index set $\mathcal{I}_{-1}=\{k:\theta_k=\pi\}$. For notational convenience, define
\begin{equation}
	\delta_k^{(-1)}
	=
	\begin{cases}
		1, & k\in\mathcal{I}_{-1},\\
		0, & k\notin\mathcal{I}_{-1},
	\end{cases}
	\qquad
	\delta_k^{(c)}=1-\delta_k^{(-1)}.
\end{equation}
Let
\begin{equation}
	\mathbf{P}_{-1}
	=
	\mathbf{V}\,\mathrm{diag}\!\bigl(\delta_1^{(-1)},\dots,\delta_N^{(-1)}\bigr)\mathbf{V}^{H}
\end{equation}
be the orthogonal projector onto the eigenspace corresponding to $\lambda=-1$, and define the complementary projector
\begin{equation}
	\mathbf{P}_c=\mathbf{I}-\mathbf{P}_{-1}.
\end{equation}
Since $\mathbf{P}_{-1}$ and $\mathbf{P}_c$ are spectral projectors of $\mathbf{F}_G$, both commute with $\mathbf{F}_G$ and with all integer powers $\mathbf{F}_G^n$.

\begin{proposition}[Exact spectral representation with isolated $\lambda=-1$ component]
	For every $\alpha\in\mathbb{R}\setminus\mathbb{Z}$,
	\begin{equation}
		\mathbf{F}_G^\alpha
		=
		e^{j\pi\alpha}\mathbf{P}_{-1}
		+
		\sum_{n=-\infty}^{\infty} c_n(\alpha)\mathbf{F}_G^n\mathbf{P}_c.
		\label{eqexact_split_representation}
	\end{equation}
\end{proposition}

\begin{proof}
	From the spectral definition,
	\begin{equation}
		\mathbf{F}_G^\alpha
		=
		\mathbf{V}\,\mathrm{diag}\!\bigl(e^{j\alpha\theta_1},\dots,e^{j\alpha\theta_N}\bigr)\mathbf{V}^{H}.
	\end{equation}
	Since $\theta_k=\pi$ for every $k\in\mathcal{I}_{-1}$, the diagonal matrix can be decomposed as
	\begin{equation}
		\begin{aligned}
			\mathrm{diag}\!\bigl(e^{j\alpha\theta_1},\dots,e^{j\alpha\theta_N}\bigr)
			&=
			e^{j\pi\alpha}\,\mathrm{diag}\!\bigl(\delta_1^{(-1)},\dots,\delta_N^{(-1)}\bigr) \\
			&\quad +
			\mathrm{diag}\!\bigl(e^{j\alpha\theta_1}\delta_1^{(c)},\dots,e^{j\alpha\theta_N}\delta_N^{(c)}\bigr).
		\end{aligned}
	\end{equation}
	Multiplying by $\mathbf{V}$ and $\mathbf{V}^H$ on the left and right gives
	\begin{equation}
		\mathbf{F}_G^\alpha
		=
		e^{j\pi\alpha}\mathbf{P}_{-1}
		+
		\mathbf{V}\,\mathrm{diag}\!\bigl(e^{j\alpha\theta_1}\delta_1^{(c)},\dots,e^{j\alpha\theta_N}\delta_N^{(c)}\bigr)\mathbf{V}^{H}.
	\end{equation}
	For every $k\notin\mathcal{I}_{-1}$, one has $\theta_k\in(-\pi,\pi)$, and hence
	\begin{equation}
		e^{j\alpha\theta_k}
		=
		\sum_{n=-\infty}^{\infty} c_n(\alpha)e^{jn\theta_k}.
	\end{equation}
	Therefore,
	\begin{equation}
		\begin{aligned}
			&\mathbf{V}\,\mathrm{diag}\!\bigl(e^{j\alpha\theta_1}\delta_1^{(c)},\dots,e^{j\alpha\theta_N}\delta_N^{(c)}\bigr)\mathbf{V}^{H} \\
			&=
			\sum_{n=-\infty}^{\infty}
			c_n(\alpha)\,
			\mathbf{V}\,\mathrm{diag}\!\bigl(e^{jn\theta_1}\delta_1^{(c)},\dots,e^{jn\theta_N}\delta_N^{(c)}\bigr)\mathbf{V}^{H}.
		\end{aligned}
	\end{equation}
	Since
	\begin{equation}
		\mathbf{F}_G^n\mathbf{P}_c
		=
		\mathbf{V}\,\mathrm{diag}\!\bigl(e^{jn\theta_1}\delta_1^{(c)},\dots,e^{jn\theta_N}\delta_N^{(c)}\bigr)\mathbf{V}^{H},
	\end{equation}
	we obtain
	\begin{equation}
		\mathbf{V}\,\mathrm{diag}\!\bigl(e^{j\alpha\theta_1}\delta_1^{(c)},\dots,e^{j\alpha\theta_N}\delta_N^{(c)}\bigr)\mathbf{V}^{H}
		=
		\sum_{n=-\infty}^{\infty} c_n(\alpha)\mathbf{F}_G^n\mathbf{P}_c.
	\end{equation}
	Combining this with the exact contribution on the $\lambda=-1$ component proves the result.
\end{proof}

Motivated by \eqref{eqexact_split_representation}, we define the truncated fast approximation.

\begin{definition}[$L$-order FGFRFT operator]
	Let $\alpha\in\mathbb{R}\setminus\mathbb{Z}$ and $L\in\mathbb{N}\cup\{0\}$. The $L$-order FGFRFT operator is defined by
	\begin{equation}
		\mathbf{Q}_L^\alpha
		=
		e^{j\pi\alpha}\mathbf{P}_{-1}
		+
		\sum_{n=-L}^{L} c_n(\alpha)\mathbf{F}_G^n\mathbf{P}_c.
		\label{eqFGFRFT_definition}
	\end{equation}
	For $\alpha\in\mathbb{Z}$, we set $\mathbf{Q}_L^\alpha=\mathbf{F}_G^\alpha$.
\end{definition}

\begin{remark}[Reduction to the standard truncated form]
	If $\mathbf{F}_G$ has no eigenvalue at $-1$, then $\mathbf{P}_{-1}=\mathbf{0}$ and $\mathbf{P}_c=\mathbf{I}$. In this case, \eqref{eqFGFRFT_definition} reduces to the standard truncated Fourier-series form
		$\mathbf{Q}_L^\alpha
		=
		\sum_{n=-L}^{L} c_n(\alpha)\mathbf{F}_G^n$.
	Therefore, the proposed definition is a strict extension of the standard construction, with the additional term only needed to treat the spectral point $\lambda=-1$ exactly.
\end{remark}

\begin{definition}[FGFRFT of a graph signal]
	For a graph signal $\mathbf{x}\in\mathbb{C}^{N}$, the FGFRFT of order $\alpha$ and truncation order $L$ is defined by
	\begin{equation}
		\widetilde{\mathbf{x}}_{\alpha}^{L}
		=
		\mathbf{Q}_L^\alpha\mathbf{x}.
	\end{equation}
\end{definition}

\begin{remark}[Why pointwise validity is sufficient]
	The matrix identity \eqref{eqexact_split_representation} does not require uniform convergence of the Fourier series on the full interval $(-\pi,\pi]$. It only requires the scalar Fourier identity to hold at the finitely many complementary spectral samples $\{\theta_{k}:k\notin\mathcal{I}_{-1}\}\subset(-\pi,\pi)$, where pointwise validity is sufficient.
\end{remark}
\subsection{Basic Properties of the FGFRFT Operator}

We next establish the algebraic properties that will be used in the reconstruction and approximation analysis.

\begin{proposition}[Efficient computation]
	For $\alpha\in\mathbb{R}\setminus\mathbb{Z}$,
\begin{equation}
	\begin{aligned}
		\mathbf{Q}_{L}^{\alpha}
		&= e^{j\pi\alpha}\mathbf{P}_{-1} + c_{0}(\alpha)\mathbf{P}_{c} \\
		&+ \sum_{n=1}^{L} \left[ c_{n}(\alpha)\mathbf{F}_{G}^{n}\mathbf{P}_{c} + c_{-n}(\alpha)(\mathbf{F}_{G}^{H})^{n}\mathbf{P}_{c} \right].
		\label{eqefficient_computation}
	\end{aligned}
\end{equation}
\end{proposition}

\begin{proof}
	Since $\mathbf{F}_{G}$ is unitary,
$\mathbf{F}_{G}^{-n}=(\mathbf{F}_{G}^{H})^{n}$, $n\ge 1$. Separating the indices $n=0$, $n>0$, and $n<0$ in \eqref{eqFGFRFT_definition}, we obtain
\begin{equation}
	\begin{aligned}
		\sum_{n=-L}^{L} c_{n}(\alpha)\mathbf{F}_{G}^{n}\mathbf{P}_{c}
		&= c_{0}(\alpha)\mathbf{P}_{c} + \sum_{n=1}^{L} \bigl[ c_{n}(\alpha)\mathbf{F}_{G}^{n}\mathbf{P}_{c} \\
		&\quad + c_{-n}(\alpha)\mathbf{F}_{G}^{-n}\mathbf{P}_{c} \bigr].
	\end{aligned}
\end{equation}

	Substituting $\mathbf{F}_{G}^{-n}=(\mathbf{F}_{G}^{H})^{n}$ yields \eqref{eqefficient_computation}.
\end{proof}

\begin{proposition}[Adjoint symmetry]
	For every $\alpha\in\mathbb{R}$ and every integer $L\ge 0$,
	\begin{equation}
		\mathbf{Q}_{L}^{-\alpha}
		=
		(\mathbf{Q}_{L}^{\alpha})^{H}.
		\label{eqadjoint_symmetry}
	\end{equation}
\end{proposition}

\begin{proof}
	If $\alpha\in\mathbb{Z}$, then $\mathbf{Q}_{L}^{\alpha}=\mathbf{F}_{G}^{\alpha}$ and
	\begin{equation}
		\mathbf{Q}_{L}^{-\alpha}
		=
		\mathbf{F}_{G}^{-\alpha}
		=
		(\mathbf{F}_{G}^{\alpha})^{H}
		=
		(\mathbf{Q}_{L}^{\alpha})^{H}.
	\end{equation}
	It remains to consider $\alpha\in\mathbb{R}\setminus\mathbb{Z}$. From the definition,
	\begin{equation}
		\mathbf{Q}_{L}^{-\alpha}
		=
		e^{-j\pi\alpha}\mathbf{P}_{-1}
		+
		\sum_{n=-L}^{L} c_{n}(-\alpha)\mathbf{F}_{G}^{n}\mathbf{P}_{c}.
	\end{equation}
	Using
	\begin{equation}
		c_{n}(-\alpha)
		=
		\frac{\sin\bigl(\pi(-\alpha-n)\bigr)}{\pi(-\alpha-n)}
		=
		\frac{\sin\bigl(\pi(\alpha+n)\bigr)}{\pi(\alpha+n)}
		=
		c_{-n}(\alpha),
	\end{equation}
	we obtain
	\begin{equation}
		\mathbf{Q}_{L}^{-\alpha}
		=
		e^{-j\pi\alpha}\mathbf{P}_{-1}
		+
		\sum_{n=-L}^{L} c_{-n}(\alpha)\mathbf{F}_{G}^{n}\mathbf{P}_{c}.
	\end{equation}
	By the change of index $m=-n$,
	\begin{equation}
		\mathbf{Q}_{L}^{-\alpha}
		=
		e^{-j\pi\alpha}\mathbf{P}_{-1}
		+
		\sum_{m=-L}^{L} c_{m}(\alpha)\mathbf{F}_{G}^{-m}\mathbf{P}_{c}.
	\end{equation}
	Since $\mathbf{P}_{c}$ is Hermitian and commutes with $\mathbf{F}_{G}^{m}$,
	\begin{equation}
		\mathbf{F}_{G}^{-m}\mathbf{P}_{c}
		=
		\bigl(\mathbf{F}_{G}^{m}\mathbf{P}_{c}\bigr)^{H}.
	\end{equation}
	Therefore,
	\begin{equation}
		\begin{aligned}
			\mathbf{Q}_{L}^{-\alpha}
			&=
			e^{-j\pi\alpha}\mathbf{P}_{-1}
			+
			\sum_{m=-L}^{L} c_{m}(\alpha)\bigl(\mathbf{F}_{G}^{m}\mathbf{P}_{c}\bigr)^{H} \\
			&=
			\left(
			e^{j\pi\alpha}\mathbf{P}_{-1}
			+
			\sum_{m=-L}^{L} c_{m}(\alpha)\mathbf{F}_{G}^{m}\mathbf{P}_{c}
			\right)^{H}
			=
			(\mathbf{Q}_{L}^{\alpha})^{H}.
		\end{aligned}
	\end{equation}
\end{proof}

\begin{proposition}[Differentiability with respect to the order]
	For $\alpha\in\mathbb{R}\setminus\mathbb{Z}$, the FGFRFT operator $\mathbf{Q}_{L}^{\alpha}$ is differentiable with respect to $\alpha$, and
	\begin{equation}
		\frac{\partial \mathbf{Q}_{L}^{\alpha}}{\partial \alpha}
		=
		j\pi e^{j\pi\alpha}\mathbf{P}_{-1}
		+
		\sum_{n=-L}^{L}
		\frac{d c_{n}(\alpha)}{d\alpha}\,
		\mathbf{F}_{G}^{n}\mathbf{P}_{c},
		\label{eqdifferentiability}
	\end{equation}
	where
	\begin{equation}
		\frac{d c_{n}(\alpha)}{d\alpha}
		=
		\frac{d}{d\alpha}
		\left(
		\frac{\sin\bigl(\pi(\alpha-n)\bigr)}{\pi(\alpha-n)}
		\right).
	\end{equation}
\end{proposition}

\begin{proof}
	For fixed $L$, the matrices $\mathbf{P}_{-1}$, $\mathbf{P}_{c}$, and $\mathbf{F}_{G}^{n}$ are independent of $\alpha$. Hence the dependence of $\mathbf{Q}_{L}^{\alpha}$ on $\alpha$ appears only through the scalar factor $e^{j\pi\alpha}$ and the coefficients $c_{n}(\alpha)$. Differentiating \eqref{eqFGFRFT_definition} term by term yields \eqref{eqdifferentiability}. Since the sinc function is smooth on $\mathbb{R}$, the derivative exists and is well defined.
\end{proof}

\begin{definition}[Approximate inverse FGFRFT]
	The approximate inverse FGFRFT is defined by
	\begin{equation}
		\widehat{\mathbf{x}}
		=
		\mathbf{Q}_{L}^{-\alpha}\widetilde{\mathbf{x}}_{\alpha}^{L}
		=
		(\mathbf{Q}_{L}^{\alpha})^{H}\mathbf{Q}_{L}^{\alpha}\mathbf{x},
	\end{equation}
	where the second equality follows from \eqref{eqadjoint_symmetry}.
\end{definition}

\subsection{Approximation Error Analysis}

We now derive the approximation error bounds for the non-integer case $\alpha\in\mathbb{R}\setminus\mathbb{Z}$. Since the $\lambda=-1$ component is treated exactly, the truncation error arises only from the complementary spectral samples.

\subsubsection{Auxiliary Results}

The following classical tools will be used.

\begin{lemma}[Abel's summation formula]
	Let $\{u_{n}\}$ and $\{v_{n}\}$ be sequences of complex numbers. For integers $M>L+1$,
	\begin{equation}
		\sum_{n=L+1}^{M} u_{n}v_{n}
		=
		U_{M}v_{M}
		+
		\sum_{n=L+1}^{M-1} U_{n}(v_{n}-v_{n+1}),
	\end{equation}
	where
	\begin{equation}
		U_{k}=\sum_{i=L+1}^{k}u_{i}.
	\end{equation}
\end{lemma}

\begin{lemma}[Dirichlet's test]
	The series $\sum_{n=L+1}^{\infty}u_{n}v_{n}$ converges if the partial sums $U_{k}=\sum_{i=L+1}^{k}u_{i}$ are uniformly bounded and the sequence $\{v_{n}\}_{n\ge L+1}$ is monotone decreasing to zero.
\end{lemma}

\begin{lemma}[Abel's inequality]
	If $\{v_{n}\}_{n\ge L+1}$ is positive and monotone decreasing, and the partial sums of $\{u_{n}\}$ satisfy $|U_{k}|\le K$ for all $k\ge L+1$, then
	\begin{equation}
		\left|\sum_{n=L+1}^{M}u_{n}v_{n}\right|
		\le
		Kv_{L+1}.
	\end{equation}
\end{lemma}

\begin{lemma}[Oscillatory partial-sum bound]
	For $\theta\in(-\pi,\pi)$, define
	\begin{equation}
		u_{n}=(-1)^{n}e^{jn\theta}=e^{jn(\theta+\pi)}.
	\end{equation}
	Then, for every $k\ge L+1$,
	\begin{equation}
		\left|\sum_{n=L+1}^{k}u_{n}\right|
		\le
		\frac{1}{\cos(\theta/2)}.
		\label{eqpartial_sum_bound}
	\end{equation}
\end{lemma}

\begin{proof}
	Since $\theta\in(-\pi,\pi)$, one has $\theta+\pi\in(0,2\pi)$ and therefore $e^{j(\theta+\pi)}\neq 1$. The partial sum is thus a geometric sum
	\begin{equation}
		\sum_{n=L+1}^{k}u_{n}
		=
		e^{j(L+1)(\theta+\pi)}
		\frac{1-e^{j(k-L)(\theta+\pi)}}{1-e^{j(\theta+\pi)}}.
	\end{equation}
	Taking absolute values and using $|1-e^{j\phi}|\le 2$ in the numerator, we obtain
	\begin{equation}
		\left|\sum_{n=L+1}^{k}u_{n}\right|
		\le
		\frac{2}{|1-e^{j(\theta+\pi)}|}.
	\end{equation}
	Now
	\begin{equation}
		|1-e^{j(\theta+\pi)}|
		=
		2\left|\sin\left(\frac{\theta+\pi}{2}\right)\right|
		=
		2\cos\left(\frac{\theta}{2}\right),
	\end{equation}
	and $\cos(\theta/2)>0$ for $\theta\in(-\pi,\pi)$. Substituting this identity proves \eqref{eqpartial_sum_bound}.
\end{proof}

\subsubsection{Pointwise Error Bound}

Define the scalar truncation error
\begin{equation}
	R_{L}(\theta)
	=
	e^{j\alpha\theta}
	-
	\sum_{n=-L}^{L} c_{n}(\alpha)e^{jn\theta},
	\qquad \theta\in(-\pi,\pi),
\end{equation}
and split it as
\begin{equation}
	R_{L}(\theta)
	=
	R_{L}^{+}(\theta)+R_{L}^{-}(\theta),
\end{equation}
where
\begin{equation}
	R_{L}^{+}(\theta)
	=
	\sum_{n=L+1}^{\infty} c_{n}(\alpha)e^{jn\theta},
	\quad
	R_{L}^{-}(\theta)
	=
	\sum_{n=-\infty}^{-L-1} c_{n}(\alpha)e^{jn\theta}.
\end{equation}

\begin{theorem}[Pointwise error bound]
	Let $\alpha\in\mathbb{R}\setminus\mathbb{Z}$ and let $L\in\mathbb{N}$ satisfy $L>|\alpha|$. Then, for every $\theta\in(-\pi,\pi)$,
	\begin{equation}
		|R_{L}(\theta)|
		\le
		\frac{|\sin(\pi\alpha)|}{\pi\cos(\theta/2)}
		\left(
		\frac{1}{L+1-\alpha}
		+
		\frac{1}{L+1+\alpha}
		\right).
		\label{eqpointwise_bound_precise}
	\end{equation}
	Equivalently,
	\begin{equation}
		|R_{L}(\theta)|
		\le
		\frac{2(L+1)|\sin(\pi\alpha)|}
		{\pi\cos(\theta/2)\bigl((L+1)^{2}-\alpha^{2}\bigr)}.
		\label{eqpointwise_bound_compact}
	\end{equation}
\end{theorem}

\begin{proof}
	\textit{Step 1: right-tail bound.}
	For every integer $n$,
	\begin{equation}
		\sin\bigl(\pi(\alpha-n)\bigr)
		=
		(-1)^{n}\sin(\pi\alpha),
	\end{equation}
	hence
	\begin{equation}
		c_{n}(\alpha)
		=
		\frac{(-1)^{n}\sin(\pi\alpha)}{\pi(\alpha-n)}.
	\end{equation}
	Therefore,
	\begin{equation}
		\begin{aligned}
			R_{L}^{+}(\theta)
			&=
			\frac{\sin(\pi\alpha)}{\pi}
			\sum_{n=L+1}^{\infty}
			\frac{(-1)^{n}e^{jn\theta}}{\alpha-n} \\
			&=
			-\frac{\sin(\pi\alpha)}{\pi}
			\sum_{n=L+1}^{\infty}
			\underbrace{(-1)^{n}e^{jn\theta}}_{u_{n}}
			\underbrace{\frac{1}{n-\alpha}}_{v_{n}}.
		\end{aligned}
	\end{equation}
	Because $L>|\alpha|$, one has $n-\alpha>0$ for all $n\ge L+1$. Thus $\{v_{n}\}_{n\ge L+1}$ is positive, monotone decreasing, and converges to zero. By \eqref{eqpartial_sum_bound}, the partial sums of $\{u_{n}\}$ are uniformly bounded by $\cos(\theta/2)^{-1}$. Dirichlet's test guarantees convergence, and Abel's inequality yields
	\begin{equation}
		|R_{L}^{+}(\theta)|
		\le
		\frac{|\sin(\pi\alpha)|}{\pi\cos(\theta/2)}
		\cdot
		\frac{1}{L+1-\alpha}.
		\label{eqright_tail_bound}
	\end{equation}
	
	\textit{Step 2: left-tail bound.}
	Let $m=-n$. Then
	\begin{equation}
		R_{L}^{-}(\theta)
		=
		\frac{\sin(\pi\alpha)}{\pi}
		\sum_{m=L+1}^{\infty}
		\frac{(-1)^{m}e^{-jm\theta}}{m+\alpha}.
	\end{equation}
	Again, because $L>|\alpha|$, the sequence $(m+\alpha)^{-1}$ is positive, monotone decreasing, and converges to zero. Moreover,
	\begin{equation}
		\left|
		\sum_{m=L+1}^{k} (-1)^{m}e^{-jm\theta}
		\right|
		=
		\left|
		\sum_{m=L+1}^{k} e^{jm(\pi-\theta)}
		\right|
		\le
		\frac{1}{\cos(\theta/2)}.
	\end{equation}
	Applying Abel's inequality again gives
	\begin{equation}
		|R_{L}^{-}(\theta)|
		\le
		\frac{|\sin(\pi\alpha)|}{\pi\cos(\theta/2)}
		\cdot
		\frac{1}{L+1+\alpha}.
		\label{eqleft_tail_bound}
	\end{equation}
	
	\textit{Step 3: combine the two tails.}
	By the triangle inequality and \eqref{eqright_tail_bound}--\eqref{eqleft_tail_bound},
	which proves \eqref{eqpointwise_bound_precise}. Combining the two fractions and \eqref{eqpointwise_bound_compact} follows immediately.
\end{proof}

\begin{remark}
	For fixed $\alpha$ and fixed $\theta\in(-\pi,\pi)$, \eqref{eqpointwise_bound_compact} implies
	\begin{equation}
		|R_{L}(\theta)|=O(L^{-1}),
		\qquad L\to\infty.
	\end{equation}
\end{remark}

\subsubsection{Matrix Approximation Error Bound}

Let
$\mathbf{E}_{L}
	=
	\mathbf{F}_{G}^{\alpha}-\mathbf{Q}_{L}^{\alpha}$.
If $r_{-1}=\mathrm{rank}(\mathbf{P}_{-1})=N$, then $\mathbf{F}_{G}=-\mathbf{I}$ and the approximation is exact for every $L$, i.e.,
\begin{equation}
	\mathbf{Q}_{L}^{\alpha}
	=
	e^{j\pi\alpha}\mathbf{I}
	=
	\mathbf{F}_{G}^{\alpha}.
\end{equation}
Hence only the case $r_{-1}<N$ requires further analysis. In this case, define
\begin{equation}
	\cos_{c,\min}
	=
	\min_{k\notin\mathcal{I}_{-1}} |\cos(\theta_{k}/2)| > 0.
\end{equation}

\begin{theorem}[Frobenius-norm error bound]
	Let $\alpha\in\mathbb{R}\setminus\mathbb{Z}$ and let $L>|\alpha|$. If $r_{-1}<N$, then
	\begin{equation}
		\|\mathbf{E}_{L}\|_{F}
		=
		\sqrt{
			\sum_{k\notin\mathcal{I}_{-1}}
			|R_{L}(\theta_{k})|^{2}
		},
		\label{eqfrobenius_error_exact}
	\end{equation}
	and
	\begin{equation}
		\|\mathbf{E}_{L}\|_{F}
		\le
		\frac{2(L+1)|\sin(\pi\alpha)|}
		{\pi\bigl((L+1)^{2}-\alpha^{2}\bigr)}
		\sqrt{
			\sum_{k\notin\mathcal{I}_{-1}}
			\frac{1}{\cos^{2}(\theta_{k}/2)}
		}.
		\label{eqfrobenius_error_tight}
	\end{equation}
	In particular,
	\begin{equation}
		\|\mathbf{E}_{L}\|_{F}
		\le
		\frac{2(L+1)\sqrt{N-r_{-1}}\,|\sin(\pi\alpha)|}
		{\pi\,\cos_{c,\min}\,\bigl((L+1)^{2}-\alpha^{2}\bigr)}.
		\label{eqfrobenius_error_simplified}
	\end{equation}
\end{theorem}

\begin{proof}
	From the exact representation \eqref{eqexact_split_representation} and the truncated definition \eqref{eqFGFRFT_definition},
	\begin{equation}
		\mathbf{E}_{L}
		=
		\mathbf{V}\,\mathrm{diag}\!\bigl(
		R_{L}(\theta_{1})\delta_{1}^{(c)},\dots,
		R_{L}(\theta_{N})\delta_{N}^{(c)}
		\bigr)\mathbf{V}^{H}.
	\end{equation}
	Indeed, the $\lambda=-1$ component is treated exactly in both $\mathbf{F}_{G}^{\alpha}$ and $\mathbf{Q}_{L}^{\alpha}$, so it contributes no truncation error. Since the Frobenius norm is unitarily invariant,
	\begin{equation}
		\begin{aligned}
			\|\mathbf{E}_{L}\|_{F}
			&=
			\left\|
			\mathrm{diag}\!\bigl(
			R_{L}(\theta_{1})\delta_{1}^{(c)},\dots,
			R_{L}(\theta_{N})\delta_{N}^{(c)}
			\bigr)
			\right\|_{F} \\
			&=
			\sqrt{
				\sum_{k=1}^{N}
				|R_{L}(\theta_{k})|^{2}\,|\delta_{k}^{(c)}|^{2}
			} \\
			&=
			\sqrt{
				\sum_{k\notin\mathcal{I}_{-1}}
				|R_{L}(\theta_{k})|^{2}
			},
		\end{aligned}
	\end{equation}
	which proves \eqref{eqfrobenius_error_exact}. Substituting the pointwise bound \eqref{eqpointwise_bound_compact} into each term yields
	\begin{equation}
		\begin{aligned}
			\|\mathbf{E}_{L}\|_{F}
			&\le
			\sqrt{
				\sum_{k\notin\mathcal{I}_{-1}}
				\left(
				\frac{2(L+1)|\sin(\pi\alpha)|}
				{\pi|\cos(\theta_{k}/2)|\bigl((L+1)^{2}-\alpha^{2}\bigr)}
				\right)^{2}
			} \\
			&=
			\frac{2(L+1)|\sin(\pi\alpha)|}
			{\pi\bigl((L+1)^{2}-\alpha^{2}\bigr)}
			\sqrt{
				\sum_{k\notin\mathcal{I}_{-1}}
				\frac{1}{\cos^{2}(\theta_{k}/2)}
			},
		\end{aligned}
	\end{equation}
	which proves \eqref{eqfrobenius_error_tight}. Finally, since
	\begin{equation}
		|\cos(\theta_{k}/2)|\ge \cos_{c,\min},
		\qquad k\notin\mathcal{I}_{-1},
	\end{equation}
	we have
	\begin{equation}
		\sum_{k\notin\mathcal{I}_{-1}}
		\frac{1}{\cos^{2}(\theta_{k}/2)}
		\le
		\frac{N-r_{-1}}{\cos_{c,\min}^{2}},
	\end{equation}
	which gives \eqref{eqfrobenius_error_simplified}.
\end{proof}

\subsection{Approximate Unitarity and Reconstruction Error}

We now show that the approximation error bound implies approximate unitarity of $\mathbf{Q}_{L}^{\alpha}$ and, consequently, a reconstruction guarantee for the approximate inverse FGFRFT.

\begin{proposition}[Approximate unitarity]
	Let
	\begin{equation}
		\epsilon_{L}=\|\mathbf{F}_{G}^{\alpha}-\mathbf{Q}_{L}^{\alpha}\|_{F}.
	\end{equation}
	Then
	\begin{equation}
		\|(\mathbf{Q}_{L}^{\alpha})^{H}\mathbf{Q}_{L}^{\alpha}-\mathbf{I}\|_{F}
		\le
		2\epsilon_{L}+\epsilon_{L}^{2}.
		\label{eqapprox_unitarity}
	\end{equation}
\end{proposition}

\begin{proof}
	Let
	\begin{equation}
		\mathbf{Q}_{L}^{\alpha}
		=
		\mathbf{F}_{G}^{\alpha}+\mathbf{\Delta},
		\qquad
		\|\mathbf{\Delta}\|_{F}=\epsilon_{L}.
	\end{equation}
	Since $\mathbf{F}_{G}^{\alpha}$ is unitary,
	\begin{equation}
		(\mathbf{F}_{G}^{\alpha})^{H}\mathbf{F}_{G}^{\alpha}=\mathbf{I}.
	\end{equation}
	Therefore,
	\begin{equation}
		\begin{aligned}
			(\mathbf{Q}_{L}^{\alpha})^{H}\mathbf{Q}_{L}^{\alpha}-\mathbf{I}
			&=
			(\mathbf{F}_{G}^{\alpha}+\mathbf{\Delta})^{H}
			(\mathbf{F}_{G}^{\alpha}+\mathbf{\Delta})
			-\mathbf{I} \\
			&=
			(\mathbf{F}_{G}^{\alpha})^{H}\mathbf{\Delta}
			+
			\mathbf{\Delta}^{H}\mathbf{F}_{G}^{\alpha}
			+
			\mathbf{\Delta}^{H}\mathbf{\Delta}.
		\end{aligned}
	\end{equation}
	Applying the triangle inequality together with the submultiplicative estimates
	\begin{equation}
		\|\mathbf{A}\mathbf{B}\|_{F}\le \|\mathbf{A}\|_{2}\|\mathbf{B}\|_{F},
		\qquad
		\|\mathbf{A}\mathbf{B}\|_{F}\le \|\mathbf{A}\|_{F}\|\mathbf{B}\|_{2},
	\end{equation}
	and using $\|\mathbf{F}_{G}^{\alpha}\|_{2}=1$ and $\|\mathbf{\Delta}\|_{2}\le \|\mathbf{\Delta}\|_{F}=\epsilon_{L}$, we obtain
	\begin{equation}
		\|(\mathbf{Q}_{L}^{\alpha})^{H}\mathbf{Q}_{L}^{\alpha}-\mathbf{I}\|_{F}
		\le
		\epsilon_{L}+\epsilon_{L}+\epsilon_{L}^{2}
		=
		2\epsilon_{L}+\epsilon_{L}^{2}.
	\end{equation}
	This proves \eqref{eqapprox_unitarity}.
\end{proof}

\begin{proposition}[Reconstruction error bound]
	Let
	\begin{equation}
		\widehat{\mathbf{x}}
		=
		(\mathbf{Q}_{L}^{\alpha})^{H}\mathbf{Q}_{L}^{\alpha}\mathbf{x}
	\end{equation}
	be the approximate reconstruction of $\mathbf{x}$. Then
	\begin{equation}
		\|\widehat{\mathbf{x}}-\mathbf{x}\|_{2}
		\le
		(2\epsilon_{L}+\epsilon_{L}^{2})\|\mathbf{x}\|_{2}.
		\label{eqreconstruction_error}
	\end{equation}
\end{proposition}

\begin{proof}
	By definition,
	\begin{equation}
		\widehat{\mathbf{x}}-\mathbf{x}
		=
		\bigl[(\mathbf{Q}_{L}^{\alpha})^{H}\mathbf{Q}_{L}^{\alpha}-\mathbf{I}\bigr]\mathbf{x}.
	\end{equation}
	Hence,
	\begin{equation}
		\begin{aligned}
			\|\widehat{\mathbf{x}}-\mathbf{x}\|_{2}
			&=
			\left\|
			\bigl[(\mathbf{Q}_{L}^{\alpha})^{H}\mathbf{Q}_{L}^{\alpha}-\mathbf{I}\bigr]\mathbf{x}
			\right\|_{2} \\
			&\le
			\|(\mathbf{Q}_{L}^{\alpha})^{H}\mathbf{Q}_{L}^{\alpha}-\mathbf{I}\|_{2}\,\|\mathbf{x}\|_{2} \\
			&\le
			\|(\mathbf{Q}_{L}^{\alpha})^{H}\mathbf{Q}_{L}^{\alpha}-\mathbf{I}\|_{F}\,\|\mathbf{x}\|_{2}.
		\end{aligned}
	\end{equation}
	Applying \eqref{eqapprox_unitarity} proves \eqref{eqreconstruction_error}.
\end{proof}

\subsection{Approximate Additivity}

We further show that the proposed FGFRFT approximately preserves the additive composition law of the exact GFRFT.

\begin{proposition}[Approximate additivity]
	For $\alpha,\beta\in\mathbb{R}$, define
	\begin{equation}
		\mathbf{\Xi}_{L}(\alpha,\beta)
		=
		\mathbf{Q}_{L}^{\alpha}\mathbf{Q}_{L}^{\beta}
		-
		\mathbf{Q}_{L}^{\alpha+\beta}.
		\label{eqapprox_additivity_defect}
	\end{equation}
	Let
	\begin{equation}
		\epsilon_{L}(\gamma)
		=
		\|\mathbf{F}_{G}^{\gamma}-\mathbf{Q}_{L}^{\gamma}\|_{F},
		\qquad \gamma\in\mathbb{R}.
	\end{equation}
	Then
	\begin{equation}
		\|\mathbf{\Xi}_{L}(\alpha,\beta)\|_{F}
		\le
		\epsilon_{L}(\alpha+\beta)
		+
		\epsilon_{L}(\alpha)
		+
		\epsilon_{L}(\beta)
		+
		\epsilon_{L}(\alpha)\epsilon_{L}(\beta).
		\label{eqapprox_additivity_bound}
	\end{equation}
\end{proposition}

\begin{proof}
	Let
		$\mathbf{Q}_{L}^{\alpha}
		=
		\mathbf{F}_{G}^{\alpha}+\mathbf{\Delta}_{\alpha}$,
		$\mathbf{Q}_{L}^{\beta}
		=
		\mathbf{F}_{G}^{\beta}+\mathbf{\Delta}_{\beta}$,
		$\mathbf{Q}_{L}^{\alpha+\beta}
		=
		\mathbf{F}_{G}^{\alpha+\beta}+\mathbf{\Delta}_{\alpha+\beta}$,
	where $\|\mathbf{\Delta}_{\alpha}\|_{F}=\epsilon_{L}(\alpha),
		\|\mathbf{\Delta}_{\beta}\|_{F}=\epsilon_{L}(\beta),
		\|\mathbf{\Delta}_{\alpha+\beta}\|_{F}=\epsilon_{L}(\alpha+\beta)$.
	Since the exact GFRFT satisfies index additivity,
	we have
	\begin{equation}
		\begin{aligned}
			\mathbf{\Xi}_{L}(\alpha,\beta)
			&=
			(\mathbf{F}_{G}^{\alpha}+\mathbf{\Delta}_{\alpha})
			(\mathbf{F}_{G}^{\beta}+\mathbf{\Delta}_{\beta})
			-
			(\mathbf{F}_{G}^{\alpha+\beta}+\mathbf{\Delta}_{\alpha+\beta}) \\
			&=
			\mathbf{F}_{G}^{\alpha}\mathbf{\Delta}_{\beta}
			+
			\mathbf{\Delta}_{\alpha}\mathbf{F}_{G}^{\beta}
			+
			\mathbf{\Delta}_{\alpha}\mathbf{\Delta}_{\beta}
			-
			\mathbf{\Delta}_{\alpha+\beta}.
		\end{aligned}
	\end{equation}
	Applying the triangle inequality and the same norm estimates used in the proof of Proposition~5, gives
	\begin{equation}
		\|\mathbf{\Xi}_{L}(\alpha,\beta)\|_{F}
		\le
		\epsilon_{L}(\alpha+\beta)
		+
		\epsilon_{L}(\alpha)
		+
		\epsilon_{L}(\beta)
		+
		\epsilon_{L}(\alpha)\epsilon_{L}(\beta),
	\end{equation}
	which proves \eqref{eqapprox_additivity_bound}.
\end{proof}

\begin{remark}
	Setting $\beta=-\alpha$ in \eqref{eqapprox_additivity_defect} and using \eqref{eqadjoint_symmetry}, the approximate unitarity result \eqref{eqapprox_unitarity} can be viewed as a special case of the approximate additivity error.
\end{remark}

\subsection{Complexity Analysis}

We finally compare the computational complexity of the proposed FGFRFT with that of the traditional GFRFT. The comparison concerns the construction of the full dense transform operator, rather than merely the application of the operator to a single signal. Assume the graph has $N$ nodes and the truncation order is $L$.

\subsubsection{Traditional GFRFT}

\paragraph{Offline phase}
A spectral decomposition of the unitary GFT matrix $\mathbf{F}_{G}$ is performed once to obtain $\mathbf{V}$, $\mathbf{V}^{H}$, and the diagonal eigenvalue matrix $\mathbf{\Sigma}$. This step costs $O(N^{3})$.

\paragraph{Online phase}
For a new order $\alpha$, the scalar powers in $\mathbf{\Sigma}^{\alpha}$ can be computed in $O(N)$. However, reconstructing the dense matrix
	$\mathbf{F}_{G}^{\alpha}
	=
	\mathbf{V}\mathbf{\Sigma}^{\alpha}\mathbf{V}^{H}$
still requires dense matrix-matrix multiplication, and the dominant complexity remains $O(N^{3})$.

\subsubsection{Proposed FGFRFT}

\paragraph{Offline phase}
The eigendecomposition of $\mathbf{F}_{G}$ is performed once. The projector $\mathbf{P}_{-1}$ is then formed from the spectral decomposition, and the complementary projector $\mathbf{P}_{c}=\mathbf{I}-\mathbf{P}_{-1}$ is obtained. Next, the matrices $\mathcal{C}=\{\mathbf{F}_{G}^{n}\mathbf{P}_{c}\}_{n=1}^{L}$
are recursively precomputed and cached. The offline cost is $O(LN^{3})$, and the storage cost is $O(LN^{2})$.

\paragraph{Online phase}
For each updated order $\alpha$, one evaluates the scalar coefficients $c_{n}(\alpha)$ and the scalar factor $e^{j\pi\alpha}$, and forms the linear combination. Since scalar-matrix multiplication and matrix addition each cost $O(N^{2})$, the online cost is $O(2LN^{2})$, which is significantly lower than $O(N^{3})$ when $L\ll N$.

Therefore, the proposed FGFRFT achieves acceleration through an offline-online tradeoff. Its main advantage arises in scenarios where the same graph is queried repeatedly for multiple transform orders, so that the offline cost can be amortized over many online updates.

To clearly demonstrate this, we compare the computational
complexity of the different algorithms in
Table~\ref{tabcomplexity}, compare the cost of constructing the full dense transform operator, not merely applying the transform to a single signal and provide the specific
implementation pseudocode of the proposed FGFRFT in
Algorithm~\ref{algffgfrft}, detailing the complete
steps from offline caching to online linear
combination. 

\begin{table}[htbp]
	\centering
	\caption{Computational complexity comparison for offline pre-computation and online operator construction}
	\label{tabcomplexity}
	\renewcommand{\arraystretch}{1.0}
	\resizebox{\linewidth}{!}{
		\begin{tabular}{cccc}
			\toprule[1.5pt]
			\textbf{Algorithm} & \textbf{Pre-computation} & \textbf{Order Update} & \textbf{Dominant Operation} \\
			\midrule
			\makecell{Traditional \\ GFRFT}
			&
			\makecell{Eigen-decomposition \\ $O(N^{3})$}
			&
			$O(N^{3})$
			&
			\makecell{Matrix-Matrix \\ Multiplication} \\
			\midrule
			\makecell{Proposed \\ FGFRFT}
			&
			\makecell{Projector and matrix powers \\ $O(LN^{3})$}
			&
			$O(2LN^{2})$
			&
			\makecell{Scalar-Matrix \\ Linear Combination} \\
			\bottomrule[1.5pt]
	\end{tabular}}
\end{table}

\begin{algorithm}[!t]
	\caption{Offline--online construction of the FGFRFT operator}
	\label{algffgfrft}
	\footnotesize
	\renewcommand{\baselinestretch}{0.9}
	\setlength{\algorithmicindent}{1em}
	\begin{algorithmic}[1]
		\REQUIRE Unitary GFT matrix $\mathbf{F}_G\in\mathbb{C}^{N\times N}$, transform order $\alpha\in\mathbb{R}$, truncation order $L$
		\ENSURE Approximate operator $\mathbf{Q}_L^\alpha \approx \mathbf{F}_G^\alpha$
		
		\STATE \textbf{Offline pre-computation}
		\STATE Compute $\mathbf{F}_G=\mathbf{V}\,\mathrm{diag}(e^{j\theta_1},\dots,e^{j\theta_N})\mathbf{V}^H$
		\STATE Identify $\mathcal{I}_{-1}=\{k:\theta_k=\pi\}$
		\STATE Form $\mathbf{W}$ by collecting the orthonormal eigenvectors associated with $\lambda=-1$
		\STATE Set $\mathbf{P}_{-1}\leftarrow \mathbf{W}\mathbf{W}^H$ and $\mathbf{P}_c\leftarrow \mathbf{I}-\mathbf{P}_{-1}$
		\STATE Set $\mathbf{M}_0\leftarrow \mathbf{P}_c$
		\FOR{$n=1$ to $L$}
		\STATE $\mathbf{M}_n\leftarrow \mathbf{F}_G\mathbf{M}_{n-1}$ 
		\ENDFOR
		
		\STATE \textbf{Online construction}
		\IF{$\alpha\in\mathbb{Z}$}
		\RETURN $\mathbf{F}_G^\alpha$
		\ENDIF
		\STATE Initialize $\mathbf{Q}_L^\alpha \leftarrow e^{j\pi\alpha}\mathbf{P}_{-1}+c_0(\alpha)\mathbf{P}_c$
		\FOR{$n=1$ to $L$}
		\STATE Compute $c_n(\alpha)=\dfrac{\sin(\pi(\alpha-n))}{\pi(\alpha-n)}$ and $c_{-n}(\alpha)=\dfrac{\sin(\pi(\alpha+n))}{\pi(\alpha+n)}$
		\STATE Update $\mathbf{Q}_L^\alpha \leftarrow \mathbf{Q}_L^\alpha + c_n(\alpha)\mathbf{M}_n + c_{-n}(\alpha)\mathbf{M}_n^H$
		\ENDFOR
		\RETURN $\mathbf{Q}_L^\alpha$
	\end{algorithmic}
\end{algorithm}

\section{Analysis of Approximation Accuracy and Computational Efficiency}

This section provides a systematic evaluation of the proposed FGFRFT in terms of approximation accuracy and computational efficiency. Using the exact GFRFT obtained from eigendecomposition as the baseline, we first introduce the experimental setup, performance metrics, and the theoretical predictions used throughout this section, including both the rigorous worst-case error bound and the Parseval-based typical-case decay law. We then verify the approximate uniform phase distribution of random unitary GFT matrices, which supports the use of the Parseval-based prediction in the random-unitary setting, and analyze the effect of the truncation order on approximation accuracy. Finally, we compare the online computational cost of FGFRFT with that of the exact GFRFT, study the robustness of the proposed approximation under perturbed inputs, and provide a controlled validation of the exact treatment of the spectral point $\lambda=-1$.

\subsection{Experimental Setup and Evaluation Metrics}

\subsubsection{Software and Hardware Environment}

All experiments are conducted in MATLAB R2024a on a 12th Gen Intel(R) Core(TM) i5-12600KF processor (3.70 GHz). Random seeds are fixed globally to ensure full reproducibility. For each node size $N$, a single random unitary matrix is generated via QR decomposition of a complex Gaussian random matrix and reused across all $(\alpha,L)$ combinations, so that timing and accuracy results are measured on identical instances. The standard GFRFT computed from eigendecomposition is used as the exact baseline, while the proposed FGFRFT is evaluated using the construction developed in Section~III. 

\subsubsection{Evaluation Metrics}
The approximation accuracy is evaluated using mean squared error (MSE), mean absolute error (MAE), and the normalized mean squared error (NMSE)
\begin{equation}
	\mathrm{MSE}
	=
	\frac{1}{N^{2}}\|\mathbf{E}_{L}\|_{F}^{2},
\end{equation}
\begin{equation}
	\mathrm{MAE}
	=
	\frac{1}{N^{2}}
	\sum_{i,j}\bigl|[\mathbf{E}_{L}]_{ij}\bigr|,
\end{equation}
\begin{equation}
	\mathrm{NMSE}
	=
	\frac{1}{N}\|\mathbf{E}_{L}\|_{F}^{2}.
\end{equation}

These three metrics satisfy the identities
$\mathrm{MSE}=\frac{\mathrm{NMSE}}{N}$,
and $\mathrm{MAE}\le \sqrt{\mathrm{MSE}}$,
where the latter follows from the Cauchy--Schwarz inequality. Among them, NMSE is scale-free and measures the relative approximation quality, MSE is additionally normalized by the graph size through $\mathrm{MSE}=\mathrm{NMSE}/N$ and thus exhibits explicit $N$-dependence, and MAE measures the average element-wise absolute deviation.
\subsubsection{Theoretical Predictions}

\paragraph{Worst-case upper bound}

The Frobenius-norm error bound derived in Section~III immediately implies a rigorous worst-case bound for the NMSE. If $r_{-1}=\mathrm{rank}(\mathbf{P}_{-1})=N$, the approximation is exact for every $L$, so that $\mathrm{NMSE}=0$.
Otherwise, when $r_{-1}<N$, the operator error satisfies
\begin{equation}\label{eq:worst_case_bound}
	\begin{aligned}
		\mathrm{NMSE}
		&=
		\frac{1}{N}\|\mathbf{E}_{L}\|_{F}^{2} \\
		&\le
		\frac{4(L+1)^{2}\sin^{2}(\pi\alpha)}
		{\pi^{2}\bigl((L+1)^{2}-\alpha^{2}\bigr)^{2}}
		\cdot
		\frac{N-r_{-1}}{N\,\cos_{c,\min}^{2}},
	\end{aligned}
\end{equation}
and consequently $\mathrm{MSE}=\frac{\mathrm{NMSE}}{N}$, $\mathrm{MAE}\le \sqrt{\mathrm{MSE}}$.

This bound follows directly from the FGFRFT error analysis in Section~III. Since the spectral component associated with $\lambda=-1$ is treated exactly, it contributes no truncation error; the worst-case constant is governed only by the complementary spectral samples. In practice, however, the quantity $\cos_{c,\min}$ can still be very small when some complementary eigenphases lie close to $\pm\pi$, so the bound may remain conservative for random unitary matrices.

\paragraph{Typical-case prediction via Parseval's identity}

A sharper characterization of the typical behavior can be obtained from Parseval's identity. Recall that the truncation error is contributed only by the complementary spectral samples. By Parseval's theorem,
\begin{equation}\label{eqparseval_identity}
	\frac{1}{2\pi}
	\int_{-\pi}^{\pi}
	|R_{L}(\theta)|^{2}\,d\theta
	=
	\sum_{|n|>L}|c_{n}(\alpha)|^{2}.
\end{equation}

On the other hand,
\begin{equation}\label{eqnmse_discrete}
	\mathrm{NMSE}
	=
	\frac{1}{N}
	\sum_{k\notin\mathcal{I}_{-1}}
	|R_{L}(\theta_{k})|^{2}.
\end{equation}
Hence, if the complementary eigenphase set
$\{\theta_kk\notin\mathcal{I}_{-1}\}$ is approximately uniformly distributed on $(-\pi,\pi)$---as verified empirically in Section~IV-B---then the discrete average over the complementary spectrum is well approximated by the integral in \eqref{eqparseval_identity}, and
\begin{equation}\label{eq:parseval_approx}
	\mathrm{NMSE}
	\approx
	\frac{N-r_{-1}}{N}
	\sum_{|n|>L}|c_n(\alpha)|^2.
\end{equation}
Since
	$|c_n(\alpha)|^2
	=
	\frac{\sin^2(\pi\alpha)}{\pi^2(\alpha-n)^2}$,
we have, for fixed $\alpha$ and large $L$,
	$\sum_{|n|>L}|c_n(\alpha)|^2
	\sim
	\frac{2\sin^2(\pi\alpha)}{\pi^2 L},
	L\to\infty$.
Therefore,
\begin{equation}\label{eq:parseval_prediction}
	\mathrm{NMSE}
	\approx
	\frac{N-r_{-1}}{N}\cdot
	\frac{2\sin^2(\pi\alpha)}{\pi^2 L}
	=
	O(L^{-1}).
\end{equation}
Accordingly, this predicts
$\mathrm{MSE}=O(N^{-1}L^{-1})$
and
$\mathrm{MAE}=O(N^{-1/2}L^{-1/2})$,
with constants depending on $\alpha$ and, through the factor $(N-r_{-1})/N$, on the proportion of the complementary spectral components. In the random-unitary setting considered here, the quantity $(N-r_{-1})/N$ is typically close to $1$, so \eqref{eq:parseval_prediction} is well approximated by the familiar $O(L^{-1})$ typical decay law.

The two predictions play complementary roles: \eqref{eq:worst_case_bound} provides a rigorous worst-case guarantee derived from the operator-level error bound, whereas \eqref{eq:parseval_prediction} captures the typical asymptotic behavior observed for the random unitary matrices used in our experiments.
	
	\subsection{Uniform Phase Distribution of Random Unitary Matrices}\label{sec:phase_verification}
	
	The Parseval-based prediction in~\eqref{eq:parseval_prediction} relies on the assumption that the eigenvalue phases are approximately uniformly distributed on $(-\pi,\pi)$. In the random-unitary setting considered here, we examine this assumption directly on the same matrices used in the subsequent experiments.
	
	\subsubsection{Histogram Analysis}
	
	Fig.~\ref{fig:phase_hist} shows the probability density histograms of the eigenvalue phases for $N\in\{1000,2000,3000\}$, together with the theoretical uniform density $\frac{1}{2\pi}\approx 0.159$.
	Across all three matrix sizes, the empirical density is nearly flat and remains in good overall agreement with the uniform reference. Although small local fluctuations are visible, no systematic deviation from uniformity is observed. As $N$ increases, these fluctuations become less pronounced, providing visual evidence that the phase distribution is well approximated by a uniform law on $(-\pi,\pi)$ in the present random-unitary setting.

	\begin{figure}[htbp]
		\centering
		\includegraphics[width=0.8\linewidth]{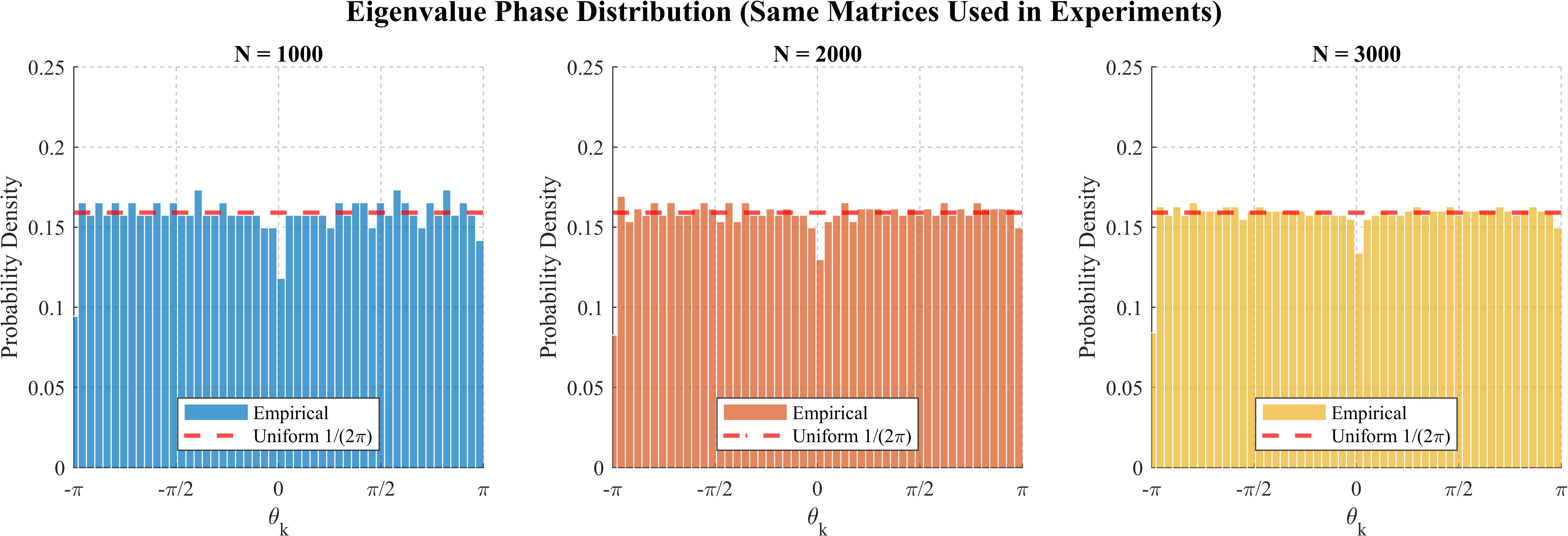}
		\caption{Probability density histograms of the eigenvalue phases $\theta_k$ of the random unitary matrices. The red dashed line indicates the density of the uniform distribution on $(-\pi,\pi)$.}
		\label{fig:phase_hist}
	\end{figure}
	
	\subsubsection{Kolmogorov--Smirnov Test}
	
	To complement the visual evidence with a statistical test, we apply the two-sided Kolmogorov--Smirnov (KS) test under the null hypothesis that the eigenvalue phases follow the distribution $\mathrm{Uniform}(-\pi,\pi)$. The KS statistic is defined by
	\begin{equation}\label{eq:ks_statistic}
		D_{KS}
		=
		\sup_{\theta}
		\left|
		\hat{F}_{N}(\theta)-F_{0}(\theta)
		\right|,
		\qquad
		F_{0}(\theta)
		=
		\frac{\theta+\pi}{2\pi},
	\end{equation}
	where $\hat{F}_{N}$ denotes the empirical cumulative distribution function (CDF).
	Fig.~\ref{fig:ks_cdf} compares the empirical CDF $\hat{F}_{N}$ with the reference CDF $F_{0}$. In all cases, the two curves are visually almost indistinguishable. The corresponding numerical results are reported in Table~\ref{tab:ks_test}.In all three cases, $D_{KS}$ is at most $17\%$ of the corresponding critical value, providing strong statistical evidence that the approximate uniform phase distribution assumption is reasonable for the random unitary matrices considered here. Together with the histogram analysis, this supports the approximation condition used in the Parseval-based analysis, under which the typical-case $O(L^{-1})$ decay law is derived.
	
	\begin{figure}[htbp]
		\centering
		\includegraphics[width=0.8\linewidth]{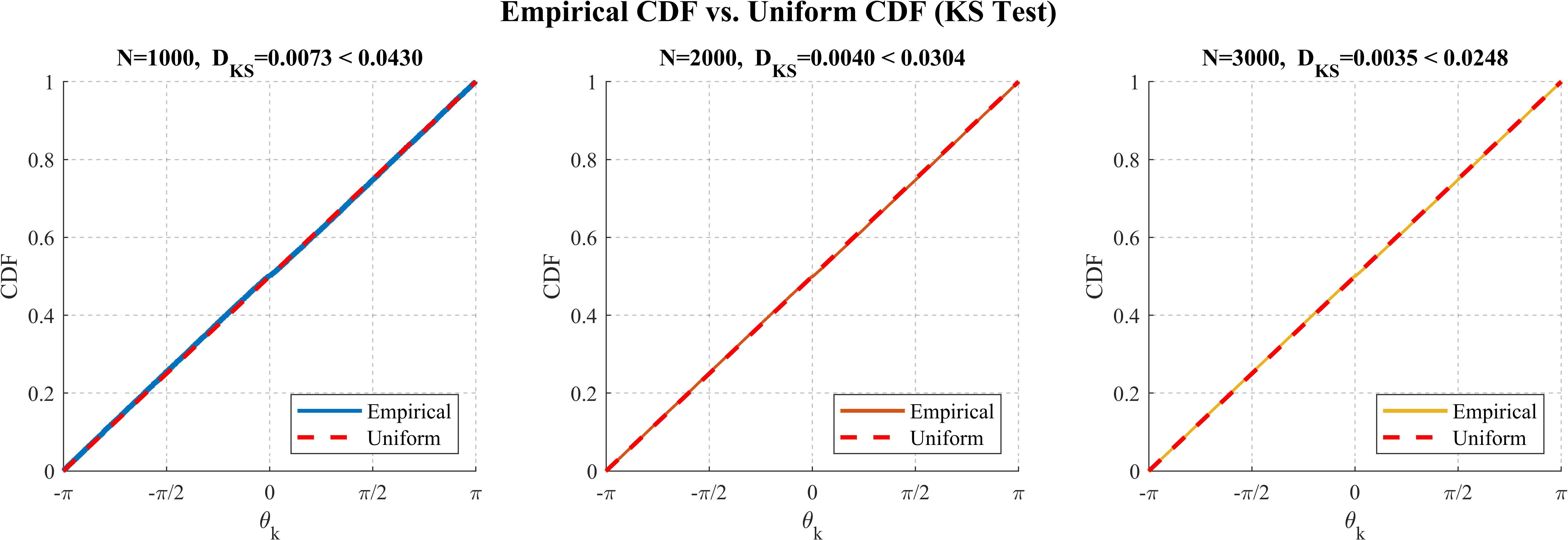}
		\caption{Empirical CDF $\hat{F}_{N}(\theta)$ versus the uniform reference CDF $F_{0}(\theta)=(\theta+\pi)/(2\pi)$. In each panel, the KS statistic $D_{KS}$ is below the $0.05$ critical value, so the null hypothesis of uniformity cannot be rejected.}
		\label{fig:ks_cdf}
	\end{figure}
	
	\begin{table}[htbp]
		\centering
		\caption{KS test results for the experimental random unitary matrices at significance level $\gamma=0.05$, with critical value $c_{\gamma}=1.36/\sqrt{N}$.}
		\label{tab:ks_test}
		\renewcommand{\arraystretch}{0.9}
		\setlength{\tabcolsep}{5pt}
		\begin{tabular}{ccccl}
			\toprule
			$N$ & $D_{KS}$ & $c_{0.05}=1.36/\sqrt{N}$ & $D_{KS}/c_{0.05}$ & Decision \\
			\midrule
			1000 & 0.0073 & 0.0430 & 0.170 & Fail to reject $H_{0}$ \\
			2000 & 0.0040 & 0.0304 & 0.132 & Fail to reject $H_{0}$ \\
			3000 & 0.0035 & 0.0248 & 0.141 & Fail to reject $H_{0}$ \\
			\bottomrule
		\end{tabular}
		\vspace{1mm}
		\footnotesize
		\begin{flushleft}
			$H_{0}$: the eigenvalue phases follow $\mathrm{Uniform}(-\pi,\pi)$. The ratio $D_{KS}/c_{0.05}\ll 1$ in all cases indicates that the null hypothesis cannot be rejected at the $0.05$ significance level.
		\end{flushleft}
	\end{table}

\subsection{Effect of Truncation Order on Approximation Accuracy}

We evaluate the approximation error for truncation orders $L \in \{10,15,20,25,30\}$, transform orders $\alpha \in \{0.15,0.35,0.\\55,0.75,0.95\}$, as well as graph sizes $N \in \{1000,2000,300\\0\}$. Fig.~\ref{fig:mse} reports the resulting MSE, MAE, and NMSE trends.

\begin{figure}[htbp]
	\centering
	\includegraphics[width=0.8\linewidth]{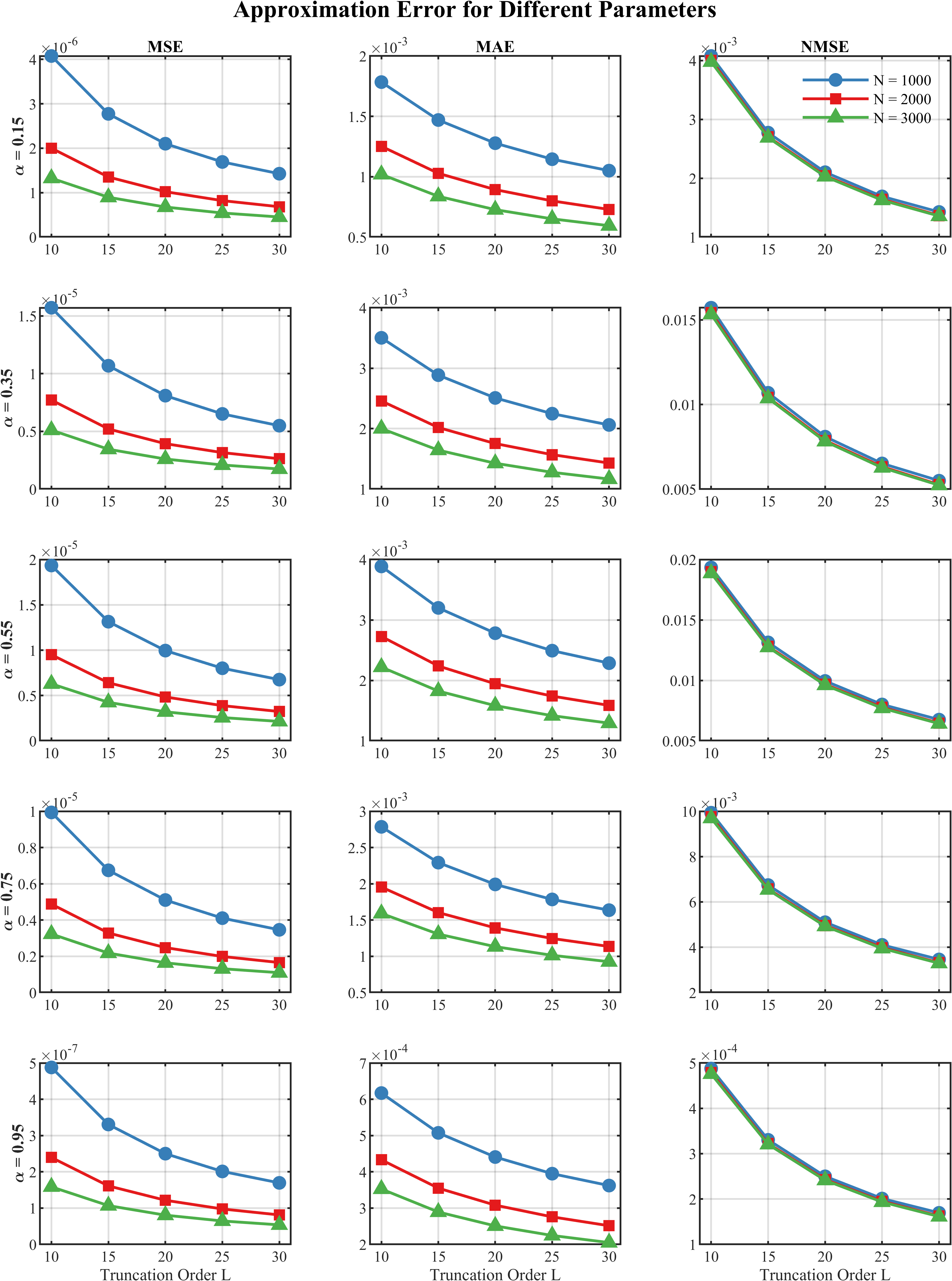}
	\caption{Approximation error versus truncation order for different graph sizes and transform orders.}
	\label{fig:mse}
\end{figure}

\subsubsection{Observed Convergence Behavior}

The following observations are drawn from Fig.~\ref{fig:mse} and are consistent with the theoretical predictions established in Section~IV-A.

\paragraph{Monotonic decrease with respect to $L$}
For all tested $(N,\alpha)$ combinations, the MSE, MAE, and NMSE decrease monotonically as $L$ increases from $10$ to $30$. This empirical behavior is consistent with the decay of the truncation error in the Fourier-series approximation.

\paragraph{Decay rate with respect to $L$}
Table~\ref{tab:rate_verification} reports $\mathrm{NMSE}\times L$ and $\mathrm{NMSE}\times L^{2}$ for the representative case $(N,\alpha)=(2000,0.55)$. The quantity $\mathrm{NMSE}\times L$ remains nearly constant ($0.190$--$0.194$), whereas $\mathrm{NMSE}\times L^{2}$ increases approximately linearly with $L$. This behavior is consistent with an $O(L^{-1})$ decay of the NMSE, in agreement with the Parseval-based prediction in~\eqref{eq:parseval_prediction}. Similar trends were observed for the other tested $(N,\alpha)$ settings. The corresponding scalings
	$\mathrm{MSE}=O(N^{-1}L^{-1})$,
	$\mathrm{MAE}=O(N^{-1/2}L^{-1/2})$,
are also consistent with the measured data.

\paragraph{Dependence on $\alpha$}
For fixed $N$ and $L$, the measured NMSE follows the $\sin^{2}(\pi\alpha)$ trend predicted by~\eqref{eq:parseval_prediction}. The error is largest near $\alpha=0.55$, where $\sin(\pi\alpha)$ is close to its maximum, and smallest near $\alpha=0.95$, where $\sin(\pi\alpha)$ is much smaller. As a representative example, at $N=2000$ and $L=10$,
\begin{equation*}
	\frac{\mathrm{NMSE}(\alpha=0.15)}
	{\mathrm{NMSE}(\alpha=0.55)}
	=
	\frac{4.0031\times 10^{-3}}{1.8995\times 10^{-2}}
	\approx 0.211,
\end{equation*}
which closely matches
	$\frac{\sin^{2}(0.15\pi)}{\sin^{2}(0.55\pi)}\approx 0.211$.
This provides a representative quantitative illustration of the $\sin^{2}(\pi\alpha)$ dependence.

\paragraph{Dependence on $N$}
For fixed $\alpha$ and $L$, the NMSE varies only mildly with $N$ over the tested range $N\in\{1000,2000,3000\}$. For example, at $\alpha=0.55$ and $L=10$, the NMSE changes from $0.01934$ ($N=1000$) to $0.018995$ ($N=2000$) and $0.018868$ ($N=3000$), remaining within a narrow range. This is consistent with the leading-order $N$-independence predicted by~\eqref{eq:parseval_prediction}. Since $\mathrm{MSE}=\mathrm{NMSE}/N$, the observed near-constancy of the NMSE implies an approximate $N^{-1}$ scaling of the MSE. For the same $(\alpha,L)$ pair, the MSE decreases from $1.934\times 10^{-5}$ ($N=1000$) to $6.2894\times 10^{-6}$ ($N=3000$), a ratio of approximately $3.08$, close to the factor $3000/1000=3$.

\paragraph{Practical accuracy}
At $L=10$, the largest NMSE observed among all tested configurations is $1.934\%$ ($N=1000$, $\alpha=0.55$). When $L$ increases to $20$, the NMSE falls below $1\%$ for all tested configurations, and at $L=30$ it is below $0.7\%$ in all cases. These results indicate that relatively small truncation orders already provide good approximation accuracy for the random-unitary setting considered here.

\begin{table}[htbp]
	\centering
	\caption{Numerical evidence for the $O(L^{-1})$ decay of NMSE for the representative case $N=2000$ and $\alpha=0.55$.}
	\label{tab:rate_verification}
	\resizebox{0.8\linewidth}{!}{
		\renewcommand{\arraystretch}{0.9}
		\begin{tabular}{cccc}
			\toprule
			$L$ & NMSE & NMSE $\times L$ & NMSE $\times L^{2}$ \\
			\midrule
			10 & 1.900E-02 & \textbf{0.190} & 1.90 \\
			15 & 1.282E-02 & \textbf{0.192} & 2.88 \\
			20 & 9.678E-03 & \textbf{0.194} & 3.87 \\
			25 & 7.773E-03 & \textbf{0.194} & 4.86 \\
			30 & 6.468E-03 & \textbf{0.194} & 5.82 \\
			\bottomrule
	\end{tabular}}
\end{table}

\subsubsection{Parseval Prediction Versus Experiment}

Before discussing the metric trends in more detail, we compare the measured NMSE with the Parseval-based prediction in~\eqref{eq:parseval_prediction}. Fig.~\ref{fig:parseval} shows three quantities for $N=2000$: (i) the measured NMSE; (ii) the exact Parseval tail sum $\sum_{|n|>L}|c_n(\alpha)|^2$; and (iii) the asymptotic approximation $\frac{2\sin^{2}(\pi\alpha)}{\pi^{2}L}$.

\begin{figure}[htbp]
	\centering
	\includegraphics[width=0.8\linewidth]{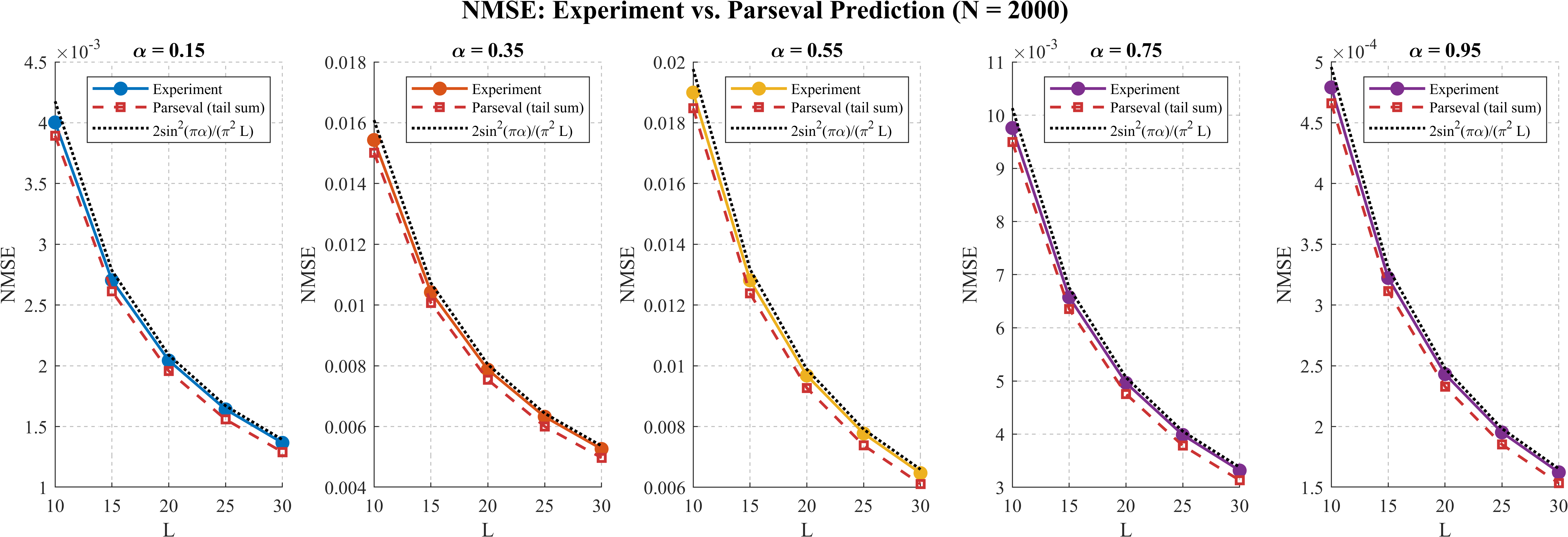}
	\caption{Measured NMSE versus Parseval-based predictions for $N=2000$. For each $\alpha$, both the exact tail sum and its asymptotic approximation closely match the measured NMSE, confirming the $O(L^{-1})$ decay.}
	\label{fig:parseval}
\end{figure}

Across all five values of $\alpha$, the three curves remain close to each other, with the maximum relative discrepancy below $5\%$ in the tested cases. This observation suggests that, once the approximate uniform phase distribution is established in Section~\ref{sec:phase_verification}, Parseval's identity provides an accurate description of both the magnitude and the $L^{-1}$ decay trend of the NMSE for the random unitary matrices considered here.

\subsubsection{Reconciling the Worst-Case Bound with the Observed Rate}\label{sec:bound_reconcile}

The rigorous worst-case bound in~\eqref{eq:worst_case_bound} scales as
$\mathrm{NMSE}
	\le
	O\!\left(
	\frac{1}{L^{2}\cos_{c,\min}^{2}}
	\right)$,
whereas the observed decay is approximately $O(L^{-1})$. These two results are not contradictory; rather, they describe different aspects of the approximation.

\paragraph{Why the worst-case bound can be loose}
The worst-case bound is obtained by applying the pointwise estimate
	$|R_L(\theta_k)|
	\le
	\frac{C}{L\,|\cos(\theta_k/2)|}$,
to each complementary eigenphase individually and then aggregating the resulting bounds. This leads to a worst-case expression controlled by the smallest complementary value
	$\cos_{c,\min}
	=
	\min_{k\notin\mathcal{I}_{-1}} |\cos(\theta_k/2)|$.
For the random unitary matrices used here, some complementary phases may still lie close to $\pi$, so $\cos_{c,\min}$ can be very small. Consequently, the prefactor in the worst-case bound can be large, making the bound rigorous but conservative in practice.

\paragraph{Why the typical decay is closer to $O(L^{-1})$}
The NMSE is an average over all contributing spectral samples, and most of the phases are well separated from $\pm\pi$. As shown in Section~\ref{sec:phase_verification}, the phase distribution is approximately uniform for the random unitary matrices considered here. In this setting, the average squared error is better characterized by Parseval's identity than by the single most unfavorable spectral sample, which explains the observed $O(L^{-1})$-type decay of
	$\sum_{|n|>L}|c_n(\alpha)|^2$.
In summary, \eqref{eq:worst_case_bound} provides a rigorous worst-case guarantee derived from the operator-level analysis, whereas Parseval's identity captures the typical behavior observed for the random unitary matrices considered in our experiments. The two viewpoints are therefore complementary rather than contradictory.

\subsubsection{Choice of Truncation Order}

Based on the typical-case approximation in~\eqref{eq:parseval_prediction}, a practical guideline for choosing the truncation order to achieve a target NMSE level $\epsilon$ is
\begin{equation}
	L
	\ge
	\left\lceil
	\frac{2\sin^{2}(\pi\alpha)}{\pi^{2}\epsilon}
	\right\rceil,
\end{equation}
where $\epsilon$ denotes the desired NMSE level. Since this rule depends on the transform order $\alpha$, it provides a refined estimate when $\alpha$ is known. For a simpler order-independent guideline, we use the bound $\sin^{2}(\pi\alpha)\le 1$, which yields
	$L
	\ge
	\left\lceil
	\frac{2}{\pi^{2}\epsilon}
	\right\rceil$.
Although this rule is conservative, it provides a convenient practical estimate that is easy to apply.

For example, the order-independent guideline gives:

	(1) $\epsilon=0.02$ ($2\%$ NMSE):
		$L \ge \left\lceil \frac{2}{\pi^{2}\times 0.02} \right\rceil
		= 11$;

	(2) $\epsilon=0.01$ ($1\%$ NMSE):
		$L \ge \left\lceil \frac{2}{\pi^{2}\times 0.01} \right\rceil
		= 21$.

These values are slightly more conservative than the empirical thresholds observed in our experiments. In particular, for the tested transform orders, $L=10$ already yields NMSE below $2\%$ in all tested configurations, while $L=20$ yields NMSE below $1\%$ in all tested configurations. Based on these observations, we use $L\in[10,20]$ in the subsequent experiments. This range provides a good accuracy--efficiency tradeoff for the tasks considered in this paper while keeping the online computational cost within $2L\in[20,40]$ scalar--matrix accumulations per order update.
	
\subsection{Comparison of Computational Efficiency on Large-Scale Graphs}\label{sec:exp_efficiency}

We evaluate the online computational time for graph sizes ranging from $N=1000$ to $N=8000$, with $L=10$ and $\alpha=0.5$ fixed. The recorded time covers only the online phase, i.e., constructing $\mathbf{Q}_{L}^{\alpha}$ for the proposed FGFRFT or reconstructing $\mathbf{F}_{G}^{\alpha}$ from cached $\mathbf{V}$, $\mathbf{\Sigma}$, and $\mathbf{V}^{H}$ for the standard GFRFT. Offline pre-computation, including eigendecomposition and the caching of matrix powers, is excluded from both timings in order to isolate the per-query cost. Therefore, the timings reported in Table~\ref{tab:efficiency_accuracy} reflect only the online operator-construction phase rather than the total runtime or memory footprint of the full pipeline.

\begin{table}[htbp]
	\centering
	\caption{Online computational time and approximation accuracy for $L=10$ and $\alpha=0.5$.}
	\label{tab:efficiency_accuracy}
	\renewcommand{\arraystretch}{0.9}
	\setlength{\tabcolsep}{5pt}
	\resizebox{\linewidth}{!}{
		\begin{tabular}{ccccccc}
			\toprule[1.5pt]
			\multirow{2}{*}{\textbf{$N$}}
			& \textbf{Standard}
			& \textbf{Proposed}
			& \textbf{Speedup}
			& \multirow{2}{*}{\textbf{MSE}}
			& \multirow{2}{*}{\textbf{MAE}}
			& \multirow{2}{*}{\textbf{NMSE}} \\
			& \textbf{GFRFT (s)}
			& \textbf{FGFRFT (s)}
			& \textbf{($\times$)} & & & \\
			\midrule[1pt]
			1000 & 0.0130 & 0.0068 & 1.92 & 2.05E-05 & 4.00E-03 & 2.05E-02 \\
			2000 & 0.0866 & 0.0366 & 2.37 & 1.00E-05 & 2.80E-03 & 2.01E-02 \\
			3000 & 0.3245 & 0.0845 & 3.84 & 6.58E-06 & 2.27E-03 & 1.97E-02 \\
			4000 & 0.7227 & 0.1586 & 4.56 & 4.92E-06 & 1.97E-03 & 1.97E-02 \\
			5000 & 1.3942 & 0.2484 & 5.61 & 3.98E-06 & 1.77E-03 & 1.99E-02 \\
			6000 & 2.3962 & 0.3582 & 6.69 & 3.34E-06 & 1.62E-03 & 2.00E-02 \\
			7000 & 3.7579 & 0.4952 & \textbf{7.59} & 2.83E-06 & 1.49E-03 & 1.98E-02 \\
			8000 & 5.5446 & 0.8120 & 6.83 & 2.46E-06 & 1.39E-03 & 1.97E-02 \\
			\bottomrule[1.5pt]
	\end{tabular}}
\end{table}

The proposed FGFRFT consistently outperforms the standard GFRFT in online runtime across all tested graph sizes. The speedup ratio increases with $N$, from $1.92\times$ at $N=1000$ to a peak of $7.59\times$ at $N=7000$, and then decreases slightly to $6.83\times$ at $N=8000$. This improvement is consistent with the reduction of the online computational complexity from $O(N^{3})$ for the standard GFRFT to $O(2LN^{2})$ for the proposed FGFRFT, where the dense matrix--matrix multiplication is replaced by an exact treatment of the $\lambda=-1$ component together with a sequence of scalar--matrix linear combinations over the complementary cached powers.

The slight drop in speedup at $N=8000$ is likely related to hardware-level memory effects. In particular, caching $L=10$ complex matrices of size $8000\times 8000$ requires approximately
	$10 \times 8000^{2} \times 16\,\mathrm{bytes}
	\approx 9.5\,\mathrm{GB}$
of memory, which can increase memory traffic and reduce practical efficiency on the tested platform. This effect is implementation- and hardware-dependent and does not change the underlying $O(2LN^{2})$ online complexity of the proposed method.
Overall, these results indicate that the proposed FGFRFT can maintain good approximation accuracy (NMSE $\approx 2\%$ for $L=10$) while providing substantial online runtime savings in the tested random-unitary setting. The results support the practical usefulness of FGFRFT for repeated online order-update scenarios in which the offline cost can be amortized across many online queries.	

\subsection{Robustness to Input Perturbations}

To evaluate the behavior of the truncated FGFRFT under input perturbations, we conduct a numerical study in the noisy-input setting. In this experiment, we fix a representative problem size $N=2000$ and consider three transform orders $\alpha \in \{0.15,0.55,0.95\}$, two truncation lengths $L \in \{10,20\}$, and four perturbation levels $sigma \in \{10^{-4},10^{-3},10^{-2},10^{-1}\}$. For each configuration, $10$ independent trials are performed. Given a clean input $\mathbf{x}$ and a complex Gaussian perturbation vector $\mathbf{n}$, the perturbed input is defined by
	$\mathbf{x}_{\sigma}
	=
	\mathbf{x}+\sigma\mathbf{n}$,
where $\sigma$ controls the perturbation magnitude.
We first measure the discrepancy between the exact output $\mathbf{F}_{G}^{\alpha}\mathbf{x}_{\sigma}$ and the truncated output $\mathbf{Q}_{L}^{\alpha}\mathbf{x}_{\sigma}$ by the relative output error
\begin{equation}
	\varepsilon_{\mathrm{out}}
	=
	\frac{
		\left\|
		\mathbf{Q}_{L}^{\alpha}\mathbf{x}_{\sigma}
		-
		\mathbf{F}_{G}^{\alpha}\mathbf{x}_{\sigma}
		\right\|_{2}
	}{
		\left\|
		\mathbf{F}_{G}^{\alpha}\mathbf{x}_{\sigma}
		\right\|_{2}
	}.
\end{equation}

To further assess how perturbations are propagated by the truncated operator, we also consider the perturbation-propagation gain
\begin{equation}
	\eta_{Q}
	=
	\frac{
		\left\|
		\mathbf{Q}_{L}^{\alpha}(\mathbf{x}+\sigma\mathbf{n})
		-
		\mathbf{Q}_{L}^{\alpha}\mathbf{x}
		\right\|_{2}
	}{
		\|\sigma\mathbf{n}\|_{2}
	}.
\end{equation}
In addition, we define the gain deviation
	$\Delta_{\eta}
	=
	|\eta_{Q}-\eta_{B}|$,
where $\eta_{B}$ denotes the corresponding perturbation-propagation gain of the exact operator $\mathbf{F}_{G}^{\alpha}$. Since $\mathbf{F}_{G}^{\alpha}$ is unitary in the present setting, one has
$\eta_{B}=1$. We also report the reconstruction error
\begin{equation}
	\varepsilon_{\mathrm{rec}}
	=
	\frac{
		\left\|
		(\mathbf{Q}_{L}^{\alpha})^{H}\mathbf{Q}_{L}^{\alpha}\mathbf{x}_{\sigma}
		-
		\mathbf{x}_{\sigma}
		\right\|_{2}
	}{
		\left\|
		\mathbf{x}_{\sigma}
		\right\|_{2}
	},
\end{equation}
which quantifies the deviation of the truncated transform from perfect reconstruction under perturbed inputs.

Fig.~\ref{fig:input_robustness} shows that the mean relative output error varies only mildly as $\sigma$ increases from $10^{-4}$ to $10^{-1}$, whereas the truncation length $L$ has a much stronger influence. For all tested transform orders, $L=20$ consistently yields lower errors than $L=10$. Among the three tested orders, the case $\alpha=0.55$ is the most challenging, while $\alpha=0.95$ yields the smallest errors. This indicates that, within the tested perturbation range, the approximation error is governed primarily by the truncation quality rather than by the perturbation level itself.
Table~\ref{tab:input_robustness} leads to the same conclusion at the largest perturbation level $\sigma=10^{-1}$. The perturbation-propagation gain $\eta_{Q}$ remains close to $1$ in all cases, while both $\Delta_{\eta}$ and $\varepsilon_{\mathrm{rec}}$ decrease as $L$ increases from $10$ to $20$. Therefore, the truncated FGFRFT preserves the perturbation behavior of the exact transform reasonably well, and its robustness is determined mainly by the truncation length.

\begin{figure}[!t]
	\centering
	\includegraphics[width=0.8\linewidth]{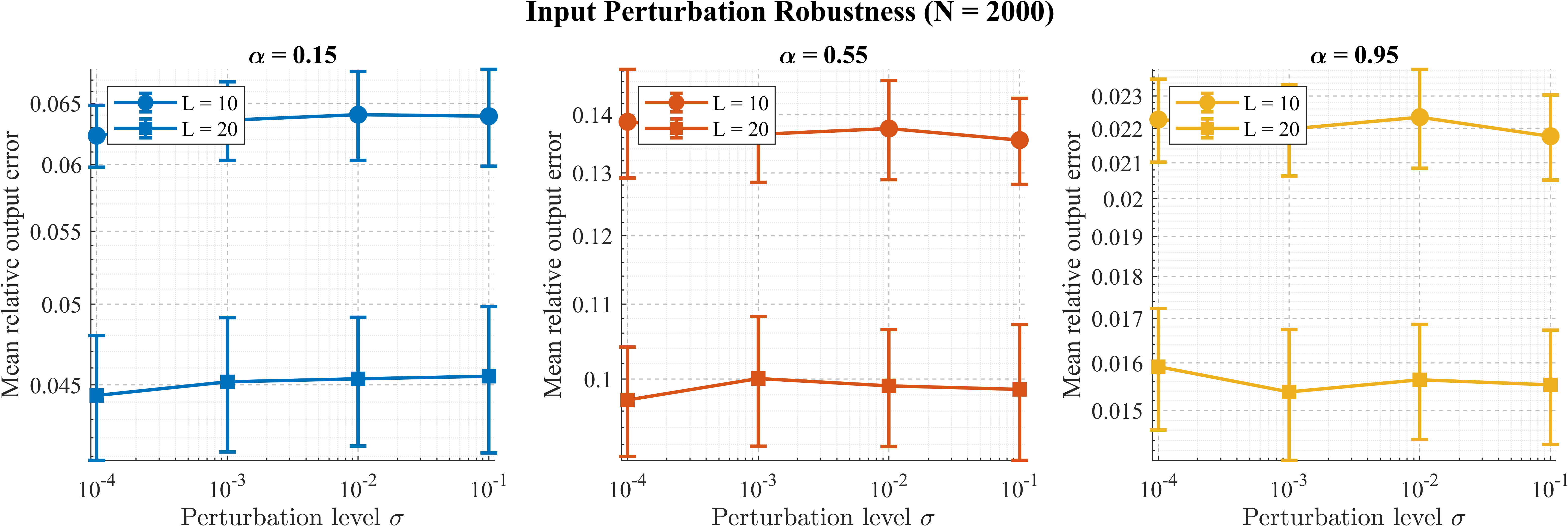}
	\caption{Mean relative output error under input perturbations for $N=2000$. Each subplot corresponds to one transform order and compares the cases $L=10$ and $L=20$. The error varies only mildly with $\sigma$ and decreases as $L$ increases.}
	\label{fig:input_robustness}
\end{figure}

\begin{table}[!t]
	\centering
	\caption{Representative robustness statistics at the largest perturbation level $\sigma=10^{-1}$ for $N=2000$.}
	\label{tab:input_robustness}
	\renewcommand{\arraystretch}{1.08}
	\resizebox{\linewidth}{!}{
		\begin{tabular}{cccccc}
			\toprule
			$\alpha$ & $L$ & Mean rel. output error & $\eta_{Q}$ & $|\eta_{Q}-\eta_{B}|$ & Reconstruction error \\
			\midrule
			0.15 & 10 & $0.0639 \pm 0.0037$ & $0.9979 \pm 0.0037$ & $0.0021 \pm 0.0004$ & $0.0378 \pm 0.0018$ \\
			0.15 & 20 & $0.0455 \pm 0.0026$ & $0.9990 \pm 0.0026$ & $0.0010 \pm 0.0003$ & $0.0279 \pm 0.0020$ \\
			0.55 & 10 & $0.1355 \pm 0.0101$ & $0.9906 \pm 0.0101$ & $0.0094 \pm 0.0021$ & $0.1894 \pm 0.0080$ \\
			0.55 & 20 & $0.0987 \pm 0.0090$ & $0.9949 \pm 0.0090$ & $0.0051 \pm 0.0020$ & $0.1370 \pm 0.0088$ \\
			0.95 & 10 & $0.0218 \pm 0.0009$ & $0.9998 \pm 0.0009$ & $0.0002 \pm 0.0001$ & $0.0124 \pm 0.0002$ \\
			0.95 & 20 & $0.0155 \pm 0.0016$ & $0.9999 \pm 0.0016$ & $0.0001 \pm 0.0001$ & $0.0070 \pm 0.0002$ \\
			\bottomrule
	\end{tabular}}
\end{table}

\subsection{Validation of the Exact Treatment of the $\lambda=-1$ Component}

To directly validate exact treatment of the spectral component associated with $\lambda=-1$, we conduct a controlled experiment using a unitary GFT matrix with a prescribed spectrum. Specifically, we construct a matrix of size $N=2000$ whose spectrum contains $r_{-1}=32$ exact eigenvalues equal to $-1$, while the remaining eigenvalues are placed on the unit circle away from $\pm\pi$. The transform order is fixed at $\alpha=0.55$, and two representative truncation orders, $L=10$ and $L=20$, are considered. We compare two operator constructions. The first is standard truncated Fourier-series form
\begin{equation}
	\mathbf{Q}_{L,\mathrm{std}}^{\alpha}
	=
	\sum_{n=-L}^{L} c_n(\alpha)\mathbf{F}_G^n,
\end{equation}
which does not isolate the $\lambda=-1$ component. The second is the proposed FGFRFT construction
\begin{equation}
	\mathbf{Q}_{L,\mathrm{prop}}^{\alpha}
	=
	e^{j\pi\alpha}\mathbf{P}_{-1}
	+
	\sum_{n=-L}^{L} c_n(\alpha)\mathbf{F}_G^n\mathbf{P}_c,
\end{equation}
which treats the $\lambda=-1$ component exactly.
To reveal the source of the approximation error, we report the Frobenius norm of the total operator error, the error restricted to the $\lambda=-1$ block, and the error restricted to the complementary block:
\begin{equation}
	E_{\mathrm{total}}
	=
	\|\mathbf{Q}_{L}^{\alpha}-\mathbf{F}_G^{\alpha}\|_F,
\end{equation}
\begin{equation}
	E_{-1}
	=
	\|\mathbf{P}_{-1}(\mathbf{Q}_{L}^{\alpha}-\mathbf{F}_G^{\alpha})\mathbf{P}_{-1}\|_F,
\end{equation}
\begin{equation}
	E_c
	=
	\|\mathbf{P}_c(\mathbf{Q}_{L}^{\alpha}-\mathbf{F}_G^{\alpha})\mathbf{P}_c\|_F.
\end{equation}

The numerical results are summarized in Table~\ref{tab:neg1_validation}. Several conclusions are immediate. First, for the standard truncated form, the dominant part of the error comes from the $\lambda=-1$ block: at both $L=10$ and $L=20$, the total error is almost entirely explained by $E_{-1}$. Second, for the proposed FGFRFT, the $\lambda=-1$ block error is reduced to numerical zero (on the order of $10^{-15}$), confirming that this spectral component is treated exactly in floating-point computation. Third, for the proposed construction, the total error coincides with the complementary-block error, which agrees precisely with the theory developed in Section~III: once the $\lambda=-1$ component is handled exactly, the remaining truncation error is contributed only by the complementary spectral components.
Moreover, as $L$ increases from $10$ to $20$, the complementary-block error decreases from $2.370129$ to $1.206450$, and the total error of the proposed FGFRFT decreases by the same amount. In contrast, the standard truncated form retains a large error floor caused by the untreated $\lambda=-1$ block. This controlled experiment therefore provides direct numerical evidence that the exact treatment of the $\lambda=-1$ component is both theoretically necessary and practically effective.

\begin{table}[htbp]
	\centering
	\caption{Validation of the exact treatment of the $\lambda=-1$ component for a controlled unitary GFT matrix with $N=2000$, $r_{-1}=32$, and $\alpha=0.55$.}
	\label{tab:neg1_validation}
	\renewcommand{\arraystretch}{1.05}
	\setlength{\tabcolsep}{6pt}
	\resizebox{0.9\linewidth}{!}{
		\begin{tabular}{ccccccc}
			\toprule
			$L$ & Construction & $E_{\mathrm{total}}$ & $E_{-1}$ & $E_c$ \\
			\midrule
			10 & Standard truncated form & $6.071996$ & $5.590316$ & $2.370129$ \\
			10 & Proposed FGFRFT         & $2.370129$ & $9.726379\times 10^{-16}$ & $2.370129$ \\
			\midrule
			20 & Standard truncated form & $5.716776$ & $5.588024$ & $1.206450$ \\
			20 & Proposed FGFRFT         & $1.206450$ & $1.223375\times 10^{-15}$ & $1.206450$ \\
			\bottomrule
	\end{tabular}}
\end{table}

\section{Experimental Results}

\subsection{Transform Order Learning Experiments}

This subsection examines whether the proposed FGFRFT
can be used as a differentiable module in end-to-end
training. We first introduce a simple multi-layer
transform network and then study its behavior under
gradient-based optimization through transform-order
learning experiments.

Consider a cascaded network with $K$ FGFRFT layers.
Each layer applies a transform operator
$\mathbf{Q}_L^{\alpha_l}$ parameterized by a learnable
order $\alpha_l$. Since $\mathbf{Q}_L^{\alpha_l}$ is
differentiable with respect to $\alpha_l$, the
transform orders can be updated by backpropagation.
Given an input graph signal $\mathbf{X}$, the network
output is
	$\hat{\mathbf{Y}}
	=
	\mathbf{Q}_L^{\alpha_K}
	\mathbf{Q}_L^{\alpha_{K-1}}
	\cdots
	\mathbf{Q}_L^{\alpha_1}\mathbf{X}$,
where $\alpha_l$ denotes the learnable transform order
of the $l$th layer. The target output is defined as
$
 \mathbf{Y} = \mathbf{F}_G^{\alpha_{\mathrm{target}}}\mathbf{X}$,
and the loss function is 
$
\mathcal{L}
=
\frac{1}{N^2}\|\hat{\mathbf{Y}}-\mathbf{Y}\|_F^2$.

In the experiments, we set the network depth to
$K \in \{1,2,3\}$, initialize each layer order to
$\alpha_{\mathrm{init}} = 0.1$, and set the target
total order to $\alpha_{\mathrm{target}} = 1.5$.
We choose $\mathbf{X}=\mathbf{I}$ so that the learning
task corresponds to approximating the target transform
operator itself. Optimization is performed using Adam
with learning rate $0.01$ for $200$ iterations.
The results are summarized in
Table~\ref{tab:learning_results}, which reports the
final loss, the learned total order
$\sum_{l=1}^{K}\alpha_l$, the total-order error
$\left|\sum_{l=1}^{K}\alpha_l-\alpha_{\mathrm{target}}\right|$,
the total training time, and the runtime speedup.
Fig.~\ref{fig:learning_curves} shows the training-loss
curves and the trajectories of the learned total order.

\begin{table}[htbp]
	\centering
	\caption{Order-learning performance for different
		network depths ($L=10, K \in \{1,2,3\}$,
		$\alpha_{\mathrm{target}} = 1.5$).}
	\label{tab:learning_results}
	\renewcommand{\arraystretch}{1.0}
	\setlength{\tabcolsep}{6pt}
	\resizebox{\linewidth}{!}{
		\begin{tabular}{ccccccc}
			\toprule[1.5pt]
			\multirow{2}{*}{\textbf{Method}} &
			\textbf{Network} &
			\textbf{Final} &
			\textbf{Final} &
			\textbf{Order} &
			\textbf{Total} &
			\textbf{Speedup} \\
			& \textbf{Depth ($K$)} &
			\textbf{Loss} &
			\textbf{$\sum \alpha$} &
			\textbf{Error} &
			\textbf{Time (s)} &
			\textbf{($\times$)} \\
			\midrule[1pt]
			
			\multirow{3}{*}{\shortstack{\textbf{Proposed}\\\textbf{FGFRFT}}}
			& 1 & 5.02E-06 & 1.5009 & 0.0009 & \textbf{436.61} & \textbf{5.15} \\
			& 2 & 7.21E-06 & 1.4989 & 0.0011 & \textbf{969.66} & \textbf{2.73} \\
			& 3 & 1.54E-05 & 1.4947 & 0.0053 & \textbf{1497.06} & \textbf{2.78} \\
			
			\midrule
			
			\multirow{3}{*}{\shortstack{Standard\\GFRFT}}
			& 1 & 4.94E-10 & 1.5009 & 0.0009 & 2250.78 & 1.00 \\
			& 2 & 4.81E-12 & 1.4999 & 0.0000 & 2650.25 & 1.00 \\
			& 3 & 1.02E-11 & 1.5000 & 0.0000 & 4163.34 & 1.00 \\
			
			\bottomrule[1.5pt]
	\end{tabular}}
	
	\vspace{1mm}
	\footnotesize
\end{table}

\begin{figure}[htbp]
	\centering
	\includegraphics[width=0.8\linewidth]{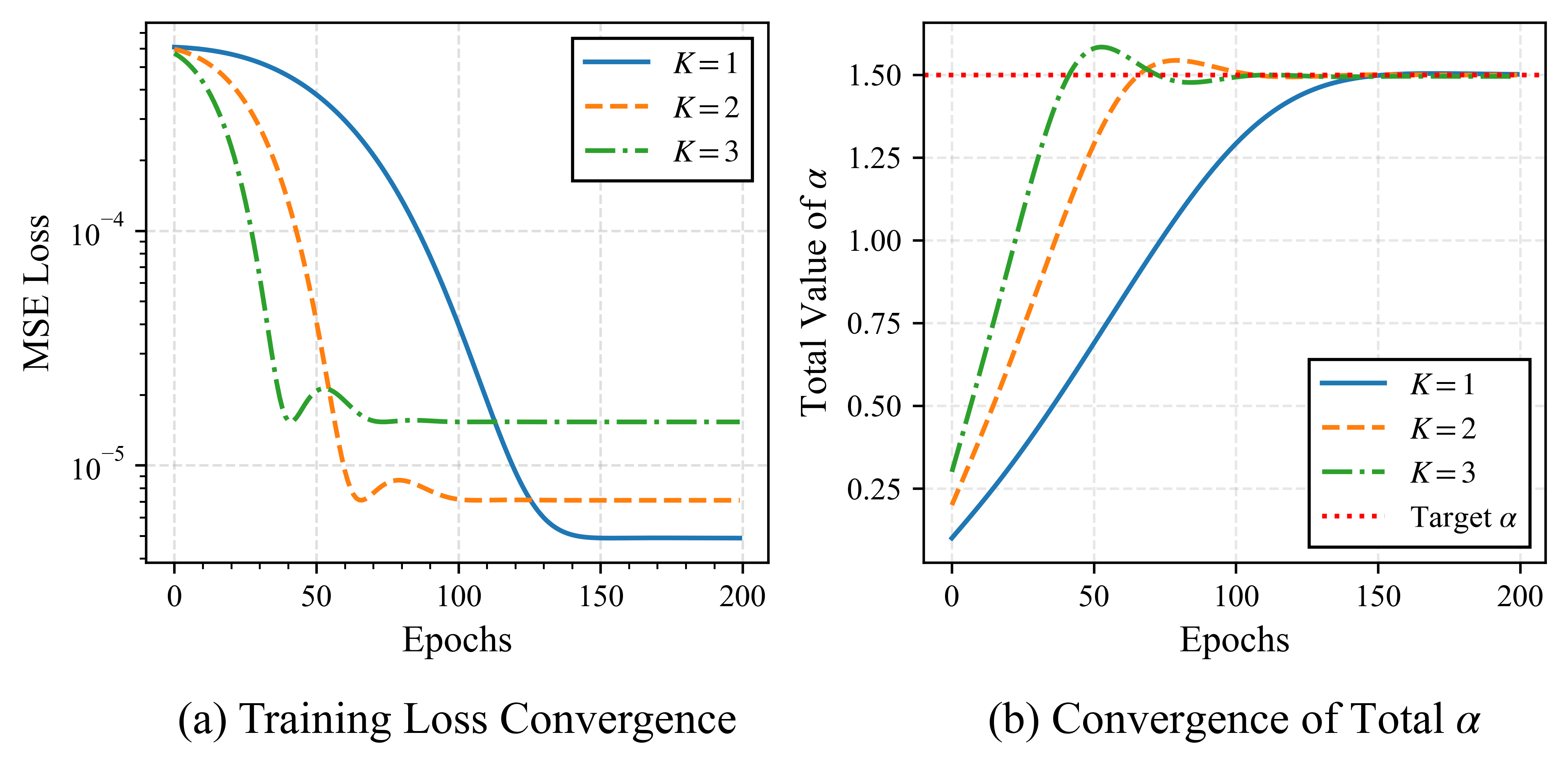}
	\caption{Training-loss curves and trajectories of the
		learned total order $\sum_l \alpha_l$ for
		different network depths.}
	\label{fig:learning_curves}
\end{figure}

Table~\ref{tab:learning_results} and
Fig.~\ref{fig:learning_curves} provide empirical
evidence that the proposed FGFRFT can be optimized by
gradient descent in this multi-layer setting. For the
tested network depths $K\in\{1,2,3\}$, the learned
total order $\sum_l \alpha_l$ approaches the target
value $\alpha_{\mathrm{target}}=1.5$ with small final
order errors ranging from $0.0009$ to $0.0053$.
This behavior is consistent with the approximate
additivity result established in Section~III and
indicates that the approximation is sufficiently
accurate for gradient-based order updates.

Compared with the standard GFRFT implementation, the
FGFRFT reaches a higher final loss, on the order of
$10^{-5}$, which is attributable to the truncation
approximation. Nevertheless, it achieves substantial
runtime reductions, with speedups ranging from
$2.73\times$ to $5.15\times$ across the tested
network depths. In addition, the training curves in
Fig.~\ref{fig:learning_curves} show that all tested
configurations converge toward the target total order
within 200 iterations, although deeper networks
exhibit slightly larger residual loss and order error.
Overall, these results suggest that the proposed
FGFRFT is trainable in an end-to-end manner and can
serve as a practical learnable transform module when a
moderate approximation error is acceptable in exchange
for improved computational efficiency.

\subsection{Image and Point-Cloud Denoising via Gradient-Based Learning}

To further assess the practical utility and
computational advantage of the proposed FGFRFT in
large-scale graph signal processing tasks, we model
images and 3D point clouds as graph signals and apply
the method to denoising problems. Following the
gradient-based denoising framework in~\cite{ref25},
we jointly optimize the transform order $\alpha$ and
the diagonal filter coefficients $\mathbf{H}$ by
backpropagation, aiming to minimize the reconstruction
error with respect to the clean signal:
\begin{equation}
(\alpha^*, \mathbf{H}^*)
=
\arg\min_{\alpha,\mathbf{H}}
\left\|
\mathbf{F}_G^{-\alpha}\mathbf{H}\mathbf{F}_G^{\alpha}\mathbf{y}
-
\mathbf{x}
\right\|_2^2.
\end{equation}

In this section, our main goal is not to redesign the
optimization strategy itself, but rather to evaluate
whether replacing the exact GFRFT operator with the
approximate FGFRFT operator under the same training
framework can reduce computation while maintaining
competitive denoising performance:
\begin{equation}
	(\alpha^*, \mathbf{H}^*)
	=
	\arg\min_{\alpha,\mathbf{H}}
	\left\|
	\mathbf{Q}_L^{-\alpha}\mathbf{H}\mathbf{Q}_L^{\alpha}\mathbf{y}
	-
	\mathbf{x}
	\right\|_2^2.
\end{equation}

\subsubsection{Image Denoising}

For image denoising, we use standard $256\times256$
images from the Set12 dataset~\cite{ref39}. To meet
the memory constraints of graph-based processing, each
image is partitioned into non-overlapping
$64\times64$ patches. Each patch is then vectorized
into a graph signal of dimension $N=4096$, and a
4-nearest-neighbor grid graph is constructed in the
pixel domain. Additive white Gaussian noise with
standard deviation $\sigma=20$ is added to the noisy
observations. We use Adam with learning rate $0.01$
for $300$ epochs. The truncation order is fixed at
$L=10$, and the initial transform order is set to
$\alpha=0.5$. For comparison, the exact GFRFT is
optimized under the same graph structure and the same
training hyperparameters. The visual results are shown
in Fig.~\ref{fig:image_denoising}, and the numerical
results are summarized in
Table~\ref{tab:image_results}.

\begin{table}[htbp]
	\centering
	\caption{Image denoising results on Set12.}
	\label{tab:image_results}
	\renewcommand{\arraystretch}{0.8}
	\setlength{\tabcolsep}{8pt}
	\resizebox{\linewidth}{!}{
		\begin{tabular}{cccccc}
			\toprule[1.5pt]
			\multirow{2}{*}{\textbf{Image}} &
			\multicolumn{2}{c}{\textbf{Standard GFRFT}} &
			\multicolumn{2}{c}{\textbf{Proposed FGFRFT}} &
			\multirow{2}{*}{\textbf{PSNR Diff.}} \\
			\cmidrule(lr){2-3} \cmidrule(lr){4-5}
			& \textbf{PSNR} & \textbf{SSIM}
			& \textbf{PSNR} & \textbf{SSIM} & \\
			\midrule
			Image 1 & 39.222 & 0.9503 & \textbf{40.132} & \textbf{0.9691} & +0.91 dB \\
			Image 2 & \textbf{38.499} & 0.9525 & 38.131 & \textbf{0.9598} & -0.37 dB \\
			Image 3 & 39.925 & 0.9694 & \textbf{41.517} & \textbf{0.9900} & +1.59 dB \\
			\midrule[1pt]
			\textbf{Time Metric} & \multicolumn{2}{c}{\textbf{Standard}} &
			\multicolumn{2}{c}{\textbf{Proposed}} & \textbf{Speedup} \\
			\midrule
			Avg.\ Time / Batch & \multicolumn{2}{c}{540 s} &
			\multicolumn{2}{c}{\textbf{220 s}} &
			\multirow{2}{*}{\textbf{2.47$\times$}} \\
			Total Time & \multicolumn{2}{c}{2.40 h} &
			\multicolumn{2}{c}{\textbf{0.97 h}} & \\
			\bottomrule[1.5pt]
	\end{tabular}}
\end{table}

\begin{figure}[htbp]
	\centering
	\includegraphics[width=0.8\linewidth]{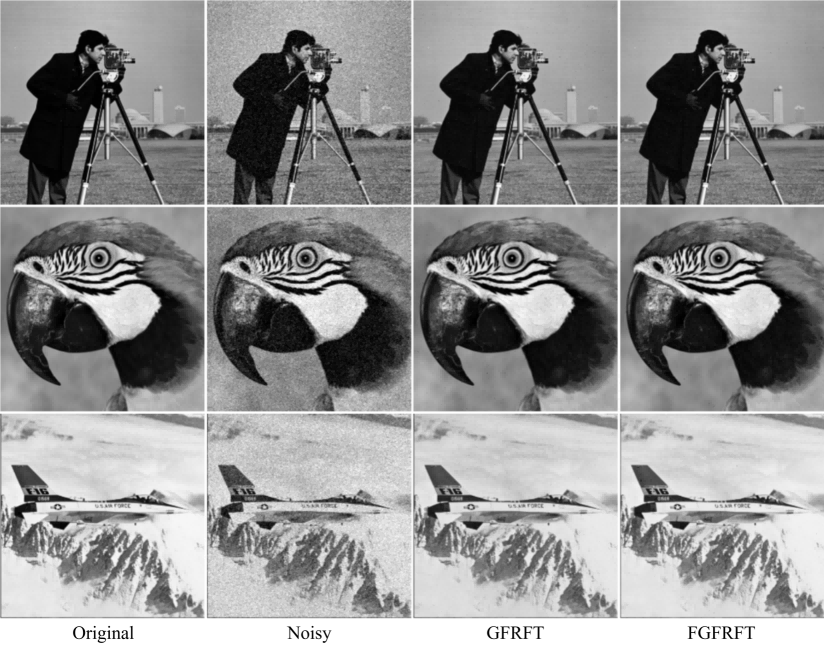}
	\caption{Visual comparison of image denoising
		results on sample images from Set12.}
	\label{fig:image_denoising}
\end{figure}

On the three tested image examples, the proposed
FGFRFT achieves denoising quality comparable to that
of the exact GFRFT, while substantially reducing
runtime. In terms of PSNR, FGFRFT performs better on
two images and slightly worse on one image; in terms
of SSIM, it is consistently higher on all three
examples. At the same time, the average processing
time per batch is reduced from $540$ s to $220$ s,
and the total runtime is reduced from $2.40$ h to
$0.97$ h, corresponding to a speedup of about
$2.47\times$.

\subsubsection{Point-Cloud Denoising}

For 3D point-cloud denoising, we use samples from the
Microsoft Voxelized Upper Bodies dataset~\cite{ref40}.
Following a preprocessing strategy similar to the
image experiment, each point cloud is downsampled and
split into multiple independent batches containing
$N=4000$ vertices. A 40-nearest-neighbor graph is
constructed using Euclidean distance. The experiment
is conducted under the same noise level
$\sigma=20.0$, while the number of training epochs is
increased to $1000$ to facilitate convergence. All
other hyperparameters are kept the same as in the
image experiment. The visual results are shown in
Fig.~\ref{fig:cloud_denoising}, and the quantitative
results are given in Table~\ref{tab:cloud_results}.

\begin{table}[!t]
	\caption{Point-cloud denoising results on the
		Microsoft voxelized dataset.}
	\label{tab:cloud_results}
	\renewcommand{\arraystretch}{0.8}
	\setlength{\tabcolsep}{8pt}
	\centering
	\resizebox{\linewidth}{!}{
		\begin{tabular}{cccc}
			\toprule[1.5pt]
			\multirow{2}{*}{\textbf{Point Cloud}} &
			\textbf{Standard GFRFT} &
			\textbf{Proposed FGFRFT} &
			\multirow{2}{*}{\textbf{PSNR Diff.}} \\
			\cmidrule(lr){2-3}
			& \textbf{PSNR (dB)} & \textbf{PSNR (dB)} & \\
			\midrule
			David9 & 37.10 & \textbf{40.85} & +3.75 dB \\
			Andrew9 & 38.95 & \textbf{39.26} & +0.31 dB \\
			Sarah9 & 42.63 & \textbf{42.96} & +0.33 dB \\
			\midrule[1pt]
			\textbf{Time Metric} & \textbf{Standard} &
			\textbf{Proposed} & \textbf{Speedup} \\
			\midrule
			Avg.\ Time / Batch & 3900 s & \textbf{1200 s} &
			\multirow{2}{*}{\textbf{3.25$\times$}} \\
			Total Time & 5.42 h & \textbf{1.67 h} & \\
			\bottomrule[1.5pt]
		\end{tabular}
	}
\end{table}

\begin{figure}[htbp]
	\centering
	\includegraphics[width=0.8\linewidth]{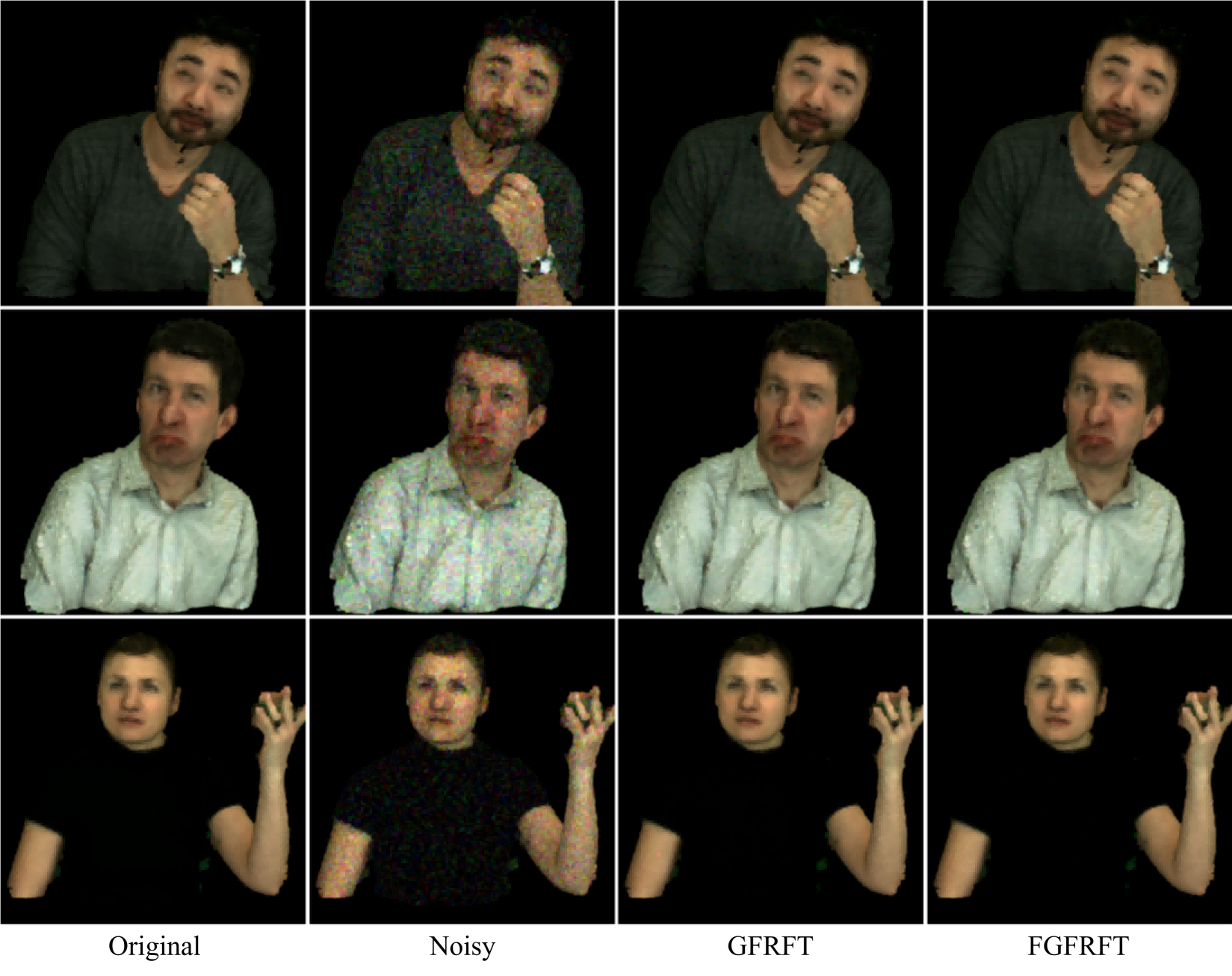}
	\caption{Visual comparison of point-cloud denoising
		results.}
	\label{fig:cloud_denoising}
\end{figure}

On the tested point-cloud examples, the proposed
FGFRFT yields PSNR values that are consistently no
worse than those of the exact GFRFT, with the largest
gain observed on David9. Meanwhile, the average
processing time per batch is reduced from $3900$ s to
$1200$ s, and the total runtime is reduced from
$5.42$ h to $1.67$ h, corresponding to a speedup of
approximately $3.25\times$.

Overall, these experiments show that replacing the
exact GFRFT with the proposed FGFRFT can
substantially reduce computational cost while
maintaining competitive denoising performance on the
tested image and point-cloud examples. In several
cases, the approximate method attains slightly higher
PSNR or SSIM than the exact implementation, but this
should be interpreted only as an empirical observation
under the current optimization setup, rather than as
a universal performance advantage of the
approximation.

\section{Conclusion}
This paper proposes FGFRFT, a fast approximation framework for repeated GFRFT order updates under unitary GFT matrices. The proposed method combines Fourier-series truncation, exact treatment of the spectral point $\lambda=-1$, and matrix-power caching to reduce the online complexity of dense operator construction from $O(N^3)$ to $O(2LN^2)$. We provide theoretical analysis of FGFRFT, including truncation-error bounds, approximate unitarity and additivity, reconstruction-error bounds, and a characterization of the approximation accuracy--efficiency tradeoff in the random-unitary setting. Experiments show that FGFRFT achieves substantial online acceleration while maintaining competitive performance in transform-order learning, image denoising, and point-cloud denoising tasks. Future work will explore more scalable caching strategies, sparse implementations, and broader evaluations on real-world graphs.

\bibliographystyle{IEEEtran}
\balance
\bibliography{reference}

\end{document}